\documentclass[final,12pt]{clear2025} 
\title[Kernel Balancing in Tree-based Methods]{Kernel Balancing in Tree-based Methods}
\usepackage{times}
\usepackage{algorithm} 
\usepackage{multirow}
\usepackage{dsfont} 
\usepackage{mathtools} 
\usepackage{booktabs} 
\usepackage{bm} 
\usepackage{physics} 
\input{ee.sty} 
\usepackage{float}
\usepackage{comment}
\usepackage{thmtools}
\usepackage{enumitem}
\usepackage{array}
\usepackage{rotating}
\usepackage[section]{placeins}
\usepackage{subcaption}
\usepackage{dsfont}
\usepackage{xurl}
\RequirePackage{natbib}
\RequirePackage{url}
\RequirePackage{hyperref}

\usepackage{longtable} 
\newenvironment{notes}
  {\par\medskip
   \begin{minipage}{\linewidth}
   \footnotesize\emph{Notes: }}
{\end{minipage}}
\renewcommand{\proofname}{Proof}
\renewenvironment{proof}[1][\proofname]
{%
  \par\noindent
  {\bfseries\upshape #1.}\ %
}
{  \jmlrQED}

\clearauthor{
 \Name{Karolina Gliszczy\'nska-Schroeder} \Email{karolina.gliszczynska(at)vwl.uni-due.de}\\
 \addr Chair of Econometrics, Faculty of Business Administration and Economics, University of Duisburg-Essen, Universitätsstraße 12, 45117 Essen, Germany%
}

\begin{document}
\newtheorem{assum}{Assumption}
\maketitle


\begin{abstract}%
Studying heterogeneous treatment effects has become essential in experimental and observational studies. A critical assumption for obtaining reliable treatment effect estimates is overlap, which requires that treated and control units have sufficiently similar covariate distributions. Poor overlap may limit the effectiveness of estimators, especially those based on propensity scores, potentially leading to unreliable results. We investigate the effectiveness of kernel balancing (KBal) \citep{hazlett2016kernel} as an alternative to propensity score methods for conditional average treatment effect (CATE) estimation, particularly in settings with overlap violations. Building on optimization-based balancing approaches, we integrate KBal weights into tree-based methods, specifically, causal forests \citep{athey_generalized_2019} and the X-Learner (XRF) \citep{kunzel_metalearners_2019}, to assess their impact on bias reduction and estimation precision. Monte Carlo evidence shows that KBal achieves near-exact balance in a transformed feature space, thereby improving treatment effect estimation in cases where traditional reweighting methods struggle due to extreme weights, finite-sample bias, or insufficient removal of pre-existing confounding bias.
We apply the proposed methods to the semi-synthetic IHDP benchmark dataset. Overall, the results indicate that KBal leads to performance improvements, especially in settings with nonlinear treatment effects and limited overlap, making it a useful alternative to propensity score methods. 

\end{abstract}

\begin{keywords}%
    Heterogeneous treatment effects, causal forests, covariate balance, weighting
\end{keywords}

\begin{jelclass}
    C14, C21, C52, C62
\end{jelclass}

\section{Introduction}
\setcounter{footnote}{0}

To obtain a valid estimator of a causal effect in the absence of randomized controlled and natural experiments, it is essential to balance the distributions of observable covariates between the treatment and control group \citep{cousineau_estimating_2022}. Various methods exist to achieve this balance,  including optimization-based approaches recently proposed in the causal inference literature \citep{Wang_2019,Zhao2019,hazlett2016kernel, wong_kernel_based_2017}. 
However, most optimization-based methods primarily focus on deriving weights that balance covariates, with little emphasis on how these weights are subsequently used for causal effect estimation. Typically, the weights are used in a simple (weighted) difference-in-means estimator, which does not directly address treatment-effect heterogeneity. When used in this way, \citet{cousineau_estimating_2022} have empirically demonstrated that optimization-based methods perform worse than regression-based or mixed methods that combine optimization-based weighting with regression adjustment in terms of estimation accuracy. This suggests that optimization-based weights may be most useful not as stand-alone estimators, but as components of more flexible causal learning procedures.  Generally, \citet{cousineau_estimating_2022} note that there is still limited awareness of how these methods can be effectively integrated with other causal inference techniques. 

Motivated by this limitation, we combine optimization-based weighting with methods for estimating conditional average treatment effects (CATE). We focus on two popular tree-based approaches: causal forests  \citep{athey_recursive_2016,  wager_estimation_2018,athey_generalized_2019} and the X-learner \citep{kunzel_metalearners_2019} with a random forest \citep{Breiman2001} as base learner. Causal forests are considered as a state-of-the-art method for estimating heterogeneous treatment effects \citep{variable_import_CF} and are especially popular in applied research because they provide a relatively accessible plug-and-play framework \citep{Applied_research_CF}. Metalearners, including the X-learner, are similarly attractive because they decompose the CATE estimation problem into subproblems, such as outcome prediction, treatment-effect imputation, and the final aggregation of treatment-effect estimates. This allows flexible base learners such as random forests, BART \citep{Chipman_2010}, or neural networks \citep{james2021introduction} to be used. Since the introduction of the S-, T-, and X-learners, the literature has developed several related approaches, including the R-learner \citep{nie2021quasi}, and the PW- and RA-learners \citep{PW_learner_RA_learner}.

In this paper, we present a perspective on integrating optimization-based weighting methods with established machine learning methods by using optimized weights as auxiliary weighting components. Specifically, we introduce two strategies. In the first approach, we use the optimized weights as sample weights in the causal forest algorithm, thereby targeting covariate imbalance within the tree-growing and aggregation process itself. In the second approach, we use these weights to replace the propensity score within the X-learner algorithm, allowing treatment effect aggregation to rely on covariate balance rather than estimated treatment assignment probabilities. Both approaches are particularly interesting in scenarios where the classical propensity score is close to 0 and 1, where inverse-probability-based methods may become unstable due to extreme weights.

The paper is structured as follows. Section \ref{sec:Potential_Outcome} reviews the potential outcome framework and outlines the role of the propensity score in causal machine learning methods and their sensitivity to extreme values. Section \ref{ch:Treebased_Kernel} introduces kernel balancing (KBal) as a special case of optimization-based methods, which is then implemented in tree-based methods. Results of a simulation study are presented in Section \ref{ch:KBAL_sim_study}, comparing the performance of the KBal variants and evaluating the coverage of confidence intervals around CATE estimates. Section \ref{ch:emp_app} illustrates the application of these methods in a semi-synthetic setting, focusing on the IHDP dataset \citep{IHDP, Hill}, where the effect of home visits from doctors on cognitive outcomes in low birth-weight, premature infants is studied. Our results suggest that incorporating KBal improves estimation accuracy in settings with nonlinear treatment structure or moderate imbalance. 
\section{Foundations of Causal Inference and Treatment Effect Estimation} \label{sec:Potential_Outcome}

We use the Neyman-Rubin potential outcome framework  \citep{ Rubin01032005, imbens_rubin_2010}. Consider a sample of $N$ units, indexed by $i = 1, \ldots, N$. For each unit $i$, let $D_i \in \{0,1\}$ be the binary treatment indicator and let $X_i \in \mathbb{R}^K$ be a $K$-dimensional covariate or feature vector. The potential outcome for unit $i$ under treatment assignment $d \in \{0,1\}$ is denoted by $Y_i(d)$, where $Y_i(0)$ represents the outcome under control, and $Y_i(1)$ the outcome under treatment.
With this framework, the individual treatment effect (ITE) is
\begin{equation}\label{eq:ITE}
    \tau_i \coloneqq Y_i(1)-Y_i(0). 
\end{equation}

To distinguish between various treatment effects of interest, let $\mathcal S$ denote the observed sample of $N$ individuals, with
$\mathcal S_1\subset\mathcal S$ and $\mathcal S_0\subset\mathcal S$ representing the subsamples for the treatment and control group, respectively,
and let $N_1$ and $N_0$ denote their corresponding sample sizes.
 For conditional estimands, we introduce the notation $\mathcal S_{\mid x}$
to denote a local subset of observations whose covariates lie in a neighborhood of $x$,
and we use the shorthand $X_i=x$ to denote this local conditioning\footnote{When $X$ is continuous, conditioning on $X_i=x$ is therefore interpreted through
membership in $\mathcal S_{\mid x}$ (e.g., a partition or tree leaf).}.
 Generally, throughout the paper, lowercase letters denote generic realizations of random variables. Thus, $x$ denotes a generic value of the covariate vector $X_i$ or in estimation contexts, a test point at which a function is evaluated. Accordingly, covariate-dependent functions, such as $\tau(x)$ introduced below, are written as functions of $x$ when referring to a generic value or test point, and as functions of $X_i$ when evaluated at the observed covariates of unit $i$.
 
The potential outcome framework rests mainly on the three assumptions. The first assumption is the Stable Unit Treatment Value Assumption (SUTVA), which rules out two potential complications. First, it excludes interference between units, so that the potential outcome of unit $i$ is affected only by its own treatment status and not by the treatment assignment of other units. Second, it requires that the treatment is well defined, meaning that there are no different or hidden versions of the same treatment level.

\begin{restatable}[Stable unit treatment value assumption (SUTVA)]{assum}{sutvaAssumption}
\label{SUTVA}
\begin{align*}
\text{If } D_i = d \text{, then } Y_i(d) = Y_i^{obs} \text{ , } \forall d \in \{0, 1 \} \text{ , } \forall ~ i= 1, ..., N . 
\end{align*}
\end{restatable}
Given SUTVA, each unit is associated with two potential outcomes, $Y_i(1)$ and $Y_i(0)$, corresponding to the treated and untreated groups. However, for each unit, only the potential outcome corresponding to the realized treatment status is observed. Hence, the observed outcome can be written as
$$Y_i^{obs} = D_iY_i(1) + (1-D_i)Y_i(0).$$
The next identifying assumption is unconfoundedness. It states that, after conditioning on the observed covariates $X_i$, treatment assignment is independent of the potential outcomes.
\begin{restatable}[Unconfoundedness]{assum}{unconfoundednessAssumption}
\label{KBAL_unconfoundedness}
\begin{align*}
D_i \perp\!\!\!\perp \left(Y_i(1), Y_i(0)\right) \mid X_i.
\end{align*}
\end{restatable}
 
In perfectly randomized experiments, Assumption \ref{KBAL_unconfoundedness} is satisfied because the researcher has direct control over the assignment mechanisms. However, in observational studies, this assumption can, at best, be verified indirectly. 

Third, we assume common support to estimate treatment effects everywhere in the covariate space. Let the conditional treatment probability or propensity score be \begin{equation}
\label{propensity} 
p(x) \coloneqq Pr(D_i=1|X_i=x).
\end{equation} 

\begin{restatable}[Overlap]{assum}{overlapAssumption}
\label{overlap} 
There exists \(\epsilon>0\) such that the conditional probability of receiving treatment given pre-treatment covariates $X$ is bounded away from zero and one,  
\begin{align*}
\epsilon < p(x) < 1-\epsilon \quad \text{for all } x \text{ in the support of } X_i.
\end{align*}
\end{restatable}

Ensuring reasonable estimates of the causal effect of interest relies on the overlap assumption. This assumption implies that units with similar pre-treatment covariates have a positive probability of being observed in both the treatment and control groups. In other words, treatment and control units should be comparable in terms of their observed characteristics. In situations where overlap is lacking, causal effect estimators, particularly those involving weighting, can become unstable or highly sensitive to model specification \citep{performance_IPW_Austin}. For a more comprehensive study of the propensity score, the associated inverse probability weighting, and the challenges that arise when the overlap assumption is violated, we refer to Section \ref{sec:propensity}.

Using Assumptions \ref{SUTVA} to \ref{overlap},  
we can identify various treatment effects based on the research question of interest. The average treatment effect (ATE), denoted as $\tau_{ATE}$, is defined as the difference in expected outcomes between individuals assigned to the treatment group and individuals assigned to the control group.  
\begin{equation}\label{eq:ATE}
 \tau_{ATE}\coloneqq \mathbb{E}[Y_i(1)-Y_i(0)].  
\end{equation}

The conditional average treatment effect (CATE), denoted as $\tau(x)$, is the expected difference in outcome for an individual with a specific set of characteristics $X_i$ when they receive treatment compared to when they do not,
\begin{align}
\label{eq:KBAL_CATE}
\tau(x)\coloneqq\mathbb{E}[Y_i(1)|X_i = x] - \mathbb{E}[Y_i(0)| X_i=x]. 
\end{align}
In many applications, one might be more interested in the average treatment effect for either the treatment or the control group rather than the entire population. Let the average treatment effect on the treated (ATT) be defined as
\begin{align}
\label{eq:ATT}
\tau_{ATT} \coloneqq \mathbb{E}[Y_i(1) - Y_i(0) | D_i = 1],
\end{align}
whereas the average treatment effect on the controls (ATC) is 
\begin{align}
\label{eq:ATC}
\tau_{ATC} \coloneqq \mathbb{E}[Y_i(1) - Y_i(0) | D_i = 0].
\end{align}

In this paper, we focus on the CATE and methods primarily geared towards its estimation. However, where appropriate, we will also provide the relevant weights for other treatment effects in our analysis.

\subsection{The Role of Propensity Scores}
\label{sec:propensity}
Weighting methods balance the control and treatment group by applying weights to individuals within each group. These weights are generally obtained through the propensity score $p(x)$ from \eqref{propensity}. The role of the propensity score is highlighted by \citet{Hirano_2003_propscore}, who show that the CATE from \eqref{eq:KBAL_CATE} can be rewritten as
$$
\tau(x)
=
\mathbb{E}\left[
Y_i^{obs}
\left(
\frac{D_i}{p(x)}
-
\frac{1-D_i}{1-p(x)}
\right)
\mid X_i=x
\right].
$$

More broadly, \citet{Rosenbaum_Rubin_propensity} show that adjusting for the propensity score removes bias due to all observed covariates. Therefore, knowing $p(x)$, would yield an unbiased estimator for $\tau(x)$.  However,  since $p(x)$ is unknown in observational studies,  it is typically estimated using logistic regression or random forest \citep{Improving_Propensity_score_2009}. These estimated propensity scores are then used to estimate the causal effects of interest through weighting \citep{rosenbaum1987model,robins1994estimation} or matching \citep{Rosenbaum_Rubin_propensity,rubin_1985_propscore}. We focus on weighting-based approaches, particularly Inverse Probability Weighting (IPW) \citep{robins1994estimation}. For a comprehensive overview of alternative propensity score methods, including covariate adjustment, stratification, and matching, we refer the reader to \citet{Austin_intro_propscore}.
The IPW weights are defined as 
\begin{align*}
w_i^{ipw} = \begin{cases} 
\frac{1}{\hat{p}(X_i)},  &\text{for } D_i=1\\
 \frac{1}{1-\hat{p}(X_i)},  &\text{for } D_i=0.  \notag
 \end{cases}
 \end{align*}
The average treatment effect can then be estimated as
\begin{align}
\hat{\tau}^{ipw}_{ATE}&=  \frac{1}{|\mathcal S|} \sum_{i\in \mathcal S}(2D_i-1)w_i^{ipw} Y_i^{obs}. 
\end{align}
Analogously, a local IPW estimator for the CATE can be written as
\begin{align}
\hat{\tau}^{ipw}(x)
&=
\frac{1}{|\mathcal S_{\mid x}|}
\sum_{i\in \mathcal S_{\mid x}}
(2D_i-1)w_i^{ipw}Y_i^{obs}.
\end{align}

Weighting methods such as IPW are useful, because they transform the observed sample into a weighted pseudo-population in which treated and control units are more comparable in terms of their observed covariates \citep{cousineau_estimating_2022}. As a result, differences in weighted outcomes are less driven by pre-existing covariate imbalance and can more plausibly be interpreted as treatment effects. However, IPW suffers when Assumption \ref{overlap} is violated, that is, when some units have estimated propensity scores close to zero or one.  In such cases, the inverse-probability weights become very large, making the resulting estimates unstable and highly variable \citep{Desail5657}. 
One way to address this violation is to remove individuals from the analysis who fall outside a range of propensity scores and focus only on those with sufficient overlap. This IPW trimming can improve the performance of IPW, but the choice of the cutoff often appears to be arbitrary. A common rule of thumb suggested by \citet{CRUMPRICHARDK_2009Dwlo} is to restrict the sample to individuals with propensity scores within the interval [0.1, 0.9], ensuring “sufficient” overlap. However, this approach may discard a significant proportion of the sample, introducing selection bias. Alternative weighting methods include overlap weights \citep{overlap_weights_li}, which  reweight observations according to $\hat{p}(x)(1-\hat{p}(x))$, assigning more weight to units with moderate propensity scores and downweighting those with extreme scores, entropy balancing \citep{Hainmueller_2012_Entropy} or stabilized weights \citep{Robins_hernan_stabil_w}. While these approaches are promising alternatives to standard IPW,   we focus exclusively on IPW and IPW trimming as a benchmark for evaluating weighting methods. 

\section{Tree-Based Methods with Kernel Weighting}
\label{ch:Treebased_Kernel}
Tree-based methods like random forests \citep{Breiman2001} are a powerful and versatile class of machine learning algorithms that have gained popularity for their ability to estimate the expected outcome  conditional on features, 
\begin{equation}\label{cond_mean_func}
\mu(x) \coloneqq \mathbb{E}[Y_i|X_i=x].
\end{equation}
This estimation is achieved by partitioning the feature space into distinct rectangles and fitting a model within each region. Tree-based models can be interpreted as nearest-neighbor methods with an adaptive neighborhood metric \citep{wager_estimation_2018}. In the traditional $k$-nearest neighbors method, the aim is to locate the $k$ closest data points to a given test point $x$ based on distance metric. However, in the case of trees, all observations that fall into the same leaf $L$ as $x$ for a tree $b$, denoted as $L_b(x)$, are treated as its neighbors. 

We focus on tree-based methods based on the CART procedure by \citet{cart84}. Additionally, we investigate more advanced methods that have evolved from CART, including the causal forest \citep{athey_recursive_2016, wager_estimation_2018}, which are a special case of the generalized random forest 
\citep{athey_generalized_2019}, and the X-learner \citep{kunzel_metalearners_2019} using the random forest as the base learner. The specific utilization of weights, and therefore our contribution, depends on the chosen method. In Section \ref{sec:Weighted_CF}, we introduce the weights directly as additional sample weights within the causal forest algorithm. In contrast, in Section \ref{sec:Weighted_XF}, we use weights as an alternative to propensity scores within the X-learner algorithm. Specifically, we employ kernel balancing \citep{hazlett2016kernel} as weighting method, which is introduced in detail in Section \ref{sec:Kernel_weighting}.
In particular, we aim to investigate whether such optimization-based weighting methods offer a good alternative to traditional weighting or propensity scores.

\subsection{ Kernel Weighting}
\label{sec:Kernel_weighting}
Kernel balancing \citep{hazlett2016kernel} is a flexible non-parametric optimization method that balances the covariates $X_i$ to obtain weights without relying on the observed outcomes. Like most optimization-based balancing methods, kernel balancing obtains the balancing weights by minimizing a discrepancy measure between the weighted group $\mathcal{U}$ and the group $\mathcal{V}$, subject to constraints. While the original work primarily focused on estimating the ATT, we extend the results to estimate the ATE or CATE by choosing the right $\mathcal{U}$ and $\mathcal{V}$, such that the resulting weights aim to create “pseudo-populations” from both the treatment and control groups that resemble the covariate distribution of the overall sample.  Following the notation of \citet{cousineau_estimating_2022} and  \citet{hazlett2016kernel}, we first discuss the basic terminology of balancing, in particular how to calculate weights on two groups to achieve balance, before moving on to the kernel balancing method as a special case. 

\begin{definition} \label{defn:optimization-based weighting}
Given two sets $\mathcal{U}$ and $\mathcal{V}$, an optimization-based weighting problem has the form 
\begin{align*}
&\operatorname{min} \sum_{i \in \mathcal{U}} g(w_i)  \quad \operatorname{subject ~to } \\
& \Bigg | \sum_{i \in \mathcal{U}} w_i \phi_k(X_i) - \frac{1}{|\mathcal{V}|} \sum_{j \in \mathcal{V}} \phi_k(X_j)\Bigg| \leq \delta_k, ~ k=1,\dots,K, \quad w_i \in \mathcal{W},
\end{align*}
where $g(w_i)$ is a convex objective function of the weights $w_i, \phi_k(X_i)$ are constraint functions of the covariates $X_i$, $\delta_k$ are tolerance parameters, $\mathcal{W}$ is a constrained set of all feasible weights and $|\mathcal V|$ denotes the cardinality of $\mathcal V$.
\end{definition}
To calculate the correct weights for our chosen treatment effect, we must select the appropriate sets $\mathcal{U}$ and $\mathcal{V}$ for balancing. \citet{cousineau_estimating_2022} provide a general weighting scheme that shows how different choices of $\mathcal{U}$ and $\mathcal{V}$ correspond to different causal estimands, which is summarized in Table \ref{table:weight_sample}. We use this scheme as a guideline for constructing kernel balancing weights for the ATE and CATE. 

\FloatBarrier
\begin{table}[h]
\centering
\caption{Samples used to construct balancing weights for different causal estimands}
\label{table:weight_sample}
\small
\setlength{\tabcolsep}{18pt}
\vspace{0.2cm}

\begin{tabular}{l r r r r}
\\[-1.8ex]\hline
\hline \\[-1.8ex]
& \multicolumn{2}{c}{$w^{1}$} & \multicolumn{2}{c}{$w^{0}$} \\
\cmidrule(lr){2-3} \cmidrule(lr){4-5}
Causal effect
& $\mathcal{U}$ & $\mathcal{V}$
& $\mathcal{U}$ & $\mathcal{V}$ \\
\hline

ATT
& -- & --
& $\mathcal{S}_{0}$ & $\mathcal{S}_{1}$ \\

ATC
& $\mathcal{S}_{1}$ & $\mathcal{S}_{0}$
& -- & -- \\

ATE
& $\mathcal{S}_{1}$ & $\mathcal{S}$
& $\mathcal{S}_{0}$ & $\mathcal{S}$ \\

CATE
& $\mathcal{S}_{1}$ & $\mathcal{S}_{|x}$
& $\mathcal{S}_{0}$ & $\mathcal{S}_{|x}$ \\

\hline
\end{tabular}

\begin{flushleft}
\small
\textit{Notes:} \(\mathcal{U}\) denotes the group of units to which weights are applied, while \(\mathcal{V}\) denotes the comparison group whose covariate distribution is used as the balancing target. The weights \(w^1\) and \(w^0\) refer to weights for treated and control units, respectively.
\end{flushleft}
\end{table}
\FloatBarrier

For example, suppose we are interested in estimating the ATT. Since the ATT is the treatment effect for the treatment group, the treated units define the target covariate distribution. Hence, only the control group has to be reweighted to resemble the treatment group. In the notation of Definition~\ref{defn:optimization-based weighting}, this corresponds to choosing $\mathcal{U}= \mathcal{S}_{0}$ and $\mathcal{V}= \mathcal{S}_{1}$. Then the optimization problem in Definition~\ref{defn:optimization-based weighting} chooses weights $w_i^0$ for $i\in\mathcal S_0$ by minimizing the objective function $g(w_i)$ subject to the balancing constraints. Examples of $g(w_i)$ are discussed below, for now, it is sufficient to view it as a convex function used to obtain weights with particular attributes. Balance is achieved through constraints of the form
$$\sum_{i\in\mathcal S_0} w_i^0\phi_k(X_i)
\approx
\frac{1}{N_1}\sum_{j\in\mathcal S_1}\phi_k(X_j),
\quad k=1,\ldots,K.$$
Here, $\phi_k(X_i)$ force the covariate features to be balanced.  For example $\phi_k(X_i)=X_{ik}$ for $k=1,\ldots,K$, balances the means of the observed covariates between the weighted control group and the treatment group. More flexible choices may include higher-order terms and interactions. Kernel balancing, discussed below, extends this idea by balancing features implicitly induced by a kernel. The tolerance parameters $\delta_k$ determine how closely the weighted control group must match the treated group. Once the weights are computed, the ATT is estimated by comparing the observed mean outcome of the treatment group with the weighted mean outcome of the reweighted control group
$$\widehat{\operatorname{ATT}} = \frac{1}{N_{1}} \sum_{{i \in \mathcal{S}_{1}}} Y_i - \sum_{{i \in \mathcal{S}_{0}}} w^{0}_i Y_i.$$ 

Note that, only one optimization problem needs to be solved for the estimation of the ATT (or ATC), since only one group is reweighted. For the ATE and CATE, the general weighting scheme in Table~\ref{table:weight_sample} computes separate weights for treatment and control units, $w^1$ and $w^0$, by solving two separate problems. However, for ATE estimation, alternative one-step formulations also exist, in which the weights for both groups are obtained jointly rather than through two separate optimization problems \citep{kallus2019optimal}. We focus on the two-step approach, as it performs well in our setting.

Kernel balancing relies on a linearity assumption, introduced formally below in Assumption \ref{Linearity of Expected Outcome}. It states that the conditional expectation of the outcome can be expressed as a linear function of transformed covariates. 
 The core idea is that if $\mathbb{E}[Y_i|X_i=x]$ lies in the linear span of a rich set of covariate transformations, then kernel balancing finds weights that equalize the weighted means
of these transformations between a set $\mathcal U$ and a target set $\mathcal V$. 
\begin{restatable}[Linearity of Expected Outcome]{assum}{linearityAssumption}
\label{Linearity of Expected Outcome} 
There exist $\theta_0 \in \mathbb{R}^Q$, $\theta_1 \in \mathbb{R}^Q$ and a function $\varphi: \mathbb{R}^K \rightarrow  \mathbb{R}^Q  $, such that the expected conditional outcomes are linear in  $\varphi$, 
$$\mathbb{E}[Y_i(0)| X_i=x ] = \varphi(x)^\top \theta_0 ~ \text{ and }~ \mathbb{E}[Y_i(1)| X_i=x ] = \varphi(x)^\top \theta_1. $$
\end{restatable}
A Gaussian kernel is not necessarily required here. \citet{hazlett2016kernel} suggested that any sufficiently rich nonlinear expansion $\varphi(\cdot)$ is
good as long as the conditional expectation functions are linear in
$\varphi(\cdot)$. 

Unlike propensity score–based methods, kernel balancing does not require Assumption \ref{overlap}. Instead, it relies on Assumption \ref{SUTVA}, \ref{KBAL_unconfoundedness} and the linearity Assumption \ref{Linearity of Expected Outcome}, which reduces the risk of bias arising from misspecifying the propensity score. Nevertheless, optimization-based weights and propensity-score weights are closely related. Propensity-score methods first model the treatment assignment probability and then use inverse probabilities as weights. Optimization-based methods take the opposite perspective. They choose weights directly so that selected covariate features are balanced between groups.

Mathematically, this connection can be explained through convex optimization theory.  The original convex optimization problem is called the primal problem. A dual problem is an alternative formulation of the same optimization task, obtained by attaching Lagrange multipliers to the constraints and optimizing over these multipliers instead of the original decision variables \citep{boyd2004convex}. 
Here, Definition~\ref{defn:optimization-based weighting} is the primal problem and chooses the weights $w_i$ directly, subject to covariate balance constraints. For a class of convex balancing-weight problems, \citet{Wang_2019} show that the corresponding dual problem can be interpreted as fitting a regularized model for inverse propensity-score weights. In simplified form, the resulting weights can be written as $w_i = \rho'\!\left(\phi(X_i)^\top \lambda\right)$, where $\lambda$ denotes the vector of dual variables associated with the balance constraints, and the function $\rho'$ is determined by the chosen objective function $g(w_i)$. Thus, the weights are constructed as functions of the covariates, similar to inverse propensity-score weights, but without explicitly estimating the propensity score. This interpretation should not be understood as a one-to-one equivalence between optimization-based weights and conventional IPW weights. Rather, it shows that both approaches are mathematically connected ways of constructing weights that adjust for covariate imbalance. We can now formulate kernel balancing as a special case of Definition~\ref{defn:optimization-based weighting}.

\begin{definition} \label{defn:kernel-based weighting}
A kernel balancing optimization problem has the following form
\begin{align*}
&\operatorname{min} \left \{\sum_{i \in \mathcal{U}} w_i \log(w_i) \right \} \quad \operatorname{subject ~to } \\
& \sum_{i \in \mathcal{U}} w_i k(X_i,X_l) - \frac{1}{\big | \mathcal{V} \big |} \sum_{j \in \mathcal{V}} k(X_j,X_l) = 0, \quad \forall l \in \{1,...,N\}, \\
&w_i \geq 0 \text{ and } \sum_{i \in \mathcal{U}} w_i =1, 
\end{align*}
where $k(X_i,X_l)=\exp \left (\frac{-||X_i-X_l||^2}{b} \right)$ is a Gaussian kernel with scale parameter $b$. It 
determines how close $X_i$ and $X_l$ must be (in an Euclidean sense) to be considered similar.
\end{definition}

Note that the objective function originally comes from a maximization problem and that its choice is flexible. \citet{Hainmueller_2012_Entropy} suggested using $\sum_{i \in \mathcal{U}} -w_i \log(w_i)$ to maximize the entropy measure, but other objective functions are also possible. For instance,  \citet{owen2001empirical} suggested using $\sum_{i \in \mathcal{U}} \log(w_i)$ to maximize the empirical likelihood, while \citet{Zubizarreta2015} used minimum variance weights as an objective function. A detailed analysis of selecting the “best” objective function is beyond the scope of the present study. Furthermore, kernel balancing enforces exact balance by setting the tolerance parameters $\delta_k = 0$ in the optimization constraints. Since exact balance is often infeasible in practice, the method achieves approximate balance not by relaxing the constraints directly, but by computing a worst-case bound on the bias due to residual imbalance and minimizing this bound instead. For further details on this bias bound and its derivation, we refer to the original paper by \citet{hazlett2016kernel}.

The Gaussian kernel $k(X_i, X_j)$ and the function $\varphi(\cdot)$ from Assumption \ref{Linearity of Expected Outcome} are related, since by  \citet{mercer1909theorem}, for every PSD kernel function $k(\cdot, \cdot)$, there exists a corresponding feature mapping $\varphi(\cdot)$ such that  
\begin{equation}
    k(X_i, X_j) = \langle \varphi(X_i), \varphi(X_j) \rangle.
\end{equation}
 By evaluating the kernel function for all pairs of observations, we construct a kernel matrix that captures these similarity measures. 
\begin{definition} \label{def:kernel_mat}
The kernel matrix $K \in \mathbb{R}^{N \times N}$ is a symmetric, positive semi-definite matrix with elements
\[
K_{ij} = k(X_i, X_j). \]
For each observation $i$, let
\[K_i =
\bigl(k(X_i,X_1), k(X_i,X_2), \ldots, k(X_i,X_N)\bigr)^\top
\in \mathbb{R}^N, \]
which represents the similarity of $X_i$ to all other observations. Furthermore let $$\overline K=\frac{1}{N}\sum_{i=1}^N K_i
\in \mathbb{R}^{N \times 1},$$
denote the empirical average of $K_i$, taken over all observations.
For  the ATE estimation, the weights $w^1$ and $w^0$ from Definition \ref{defn:kernel-based weighting} achieve mean balance on $K$ if
\begin{align}
    \sum_{i \in \mathcal{S}_1} w_i^1 K_i=\bar K =\sum_{i \in \mathcal{S}_0} w_i^0 K_i. 
\end{align}
\end{definition}
We can now extend Theorem 1 from \citet{hazlett2016kernel} to establish unbiasedness of the weighted difference-in-means estimator for the ATE, and to motivate extensions to conditional estimands. 
\begin{theorem} \label{thm:DIM_estimator}
Consider the weighted difference-in-means estimator,
\[ \widehat{\text{DIM}}_w \coloneqq \sum_{i \in \mathcal{S}_1} w_i^1 Y_i - \sum_{i \in \mathcal{S}_0} w_i^0 Y_i, \]
Let $w^0$ and $w^1$ be the optimal weights from Definition \ref{defn:kernel-based weighting}. Suppose that these weights satisfy
\begin{align*}
\sum_{i \in \mathcal{S}_1} w_i^1 \varphi(X_i) &= \sum_{i \in \mathcal{S}_0} w_i^0 \varphi(X_i) = \frac{1}{N}\sum_{i=1}^N \varphi(X_i), \\
\text{with} \quad \sum_{i \in \mathcal{S}_1} w_i^1 &= \sum_{i \in \mathcal{S}_0} w_i^0 = 1, \quad w_i^1, w_i^0 > 0.
\end{align*}
Then, under Assumptions~\ref{SUTVA}, \ref{KBAL_unconfoundedness}, and \ref{Linearity of Expected Outcome}, $\widehat{\mathrm{DIM}}_w$ is unbiased for the ATE, with unbiasedness taken over the common joint distribution of $(X_i,Y_i(1),Y_i(0),D_i)$.
\end{theorem}
\begin{proof}
  See Appendix \ref{proof:DIM_estimator}.
\end{proof}
\begin{remark}
Although Theorem~\ref{thm:DIM_estimator} is stated in terms of balance on $\varphi(\cdot)$, kernel balancing enforces balance on $K_i$ by construction. As shown in \citet{hazlett2016kernel}, enforcing mean balance on $K_i$ is sufficient to eliminate bias arising from linear functions of $\varphi(X_i)$. Hence, achieving equal means on $\varphi(X_i)$
can be replaced by achieving equal means on $K_i$. 
\end{remark}

\begin{remark}
A similar result to Theorem~\ref{thm:DIM_estimator} holds for the CATE by using a covariate-dependent target distribution, for example a local subset $\mathcal{S}_{|x} \subseteq \mathcal{S}$ of observations around a test point $x$. However, we estimate CATE by combining globally computed ATE-targeting kernel weights with local neighborhood structures induced
by tree-based methods (see Section~\ref{sec:Weighted_CF} and~\ref{sec:Weighted_XF}). A formal analysis of kernel balancing with covariate-specific target distributions
is beyond the scope of this paper.
\end{remark}

\subsection{Weighted Causal Forest}
\label{sec:Weighted_CF}

We first review  the standard causal forest estimator as a special case of the generalized random forest (GRF) framework by \citet{athey_generalized_2019}. We then explain how kernel balancing weights are incorporated as sample weights to obtain the kernel-weighted causal forest used in this paper. We focus on causal forests rather than the causal tree proposed by \citet{athey_recursive_2016}. Both causal forests and causal trees are extensions of the random forest \citep{Breiman2001}, but they differ from standard random forests in several ways. Most fundamentally, the objective shifts from a prediction task to estimating a causal effect.
Furthermore, they rely on “honest” splitting, which involves dividing the data into two disjoint sets, one for constructing the trees and another for estimating treatment effects. Second, the splitting criterion is adapted to the causal setting. While the original random forests select splits to minimize the prediction error of the outcome variable, causal forests split the data to maximize heterogeneity.  In doing so, they reveal how treatment effects vary across the sample. In the causal forest algorithm proposed by \citet{athey_generalized_2019}, $\tau(x)$ is identified as the solution to a local moment condition\footnote{More generally, the GRF framework targets parameters of the form  $\theta(x) := \xi \cdot \tau(x)$, where $\xi$  is a contrast vector. For simplicity and consistency with the rest of our paper, we focus on the $\tau(x)$ component only.}
\begin{equation}
    \mathbb{E}\left[(Y_i - \tau(x) \cdot D_i - c(x))(1, D_i)^\top \;\middle|\; X_i = x \right] = 0 \quad \text{for all } x ,
\end{equation}
where $c(x)$ is an intercept term. To estimate $\tau(x)$, the causal forest algorithm uses a weighting scheme to build the final treatment effect estimate rather than having each tree compute its estimate. For this, we define similarity weights $\alpha_i(x)$ to capture the frequency with which the $i$-observation falls in the same leaf as $x$ across the $B$ trees in the forest:
\begin{align} \label{eq:alpha}
\alpha_{bi}(x) = \frac{\mathds{1}(\{X_i \in L_b(x)\})} {|L_b(x)|},   ~ \alpha_{i}(x) = \frac{1}{B} \sum_{b=1}^B \alpha_{bi}(x).
\end{align}  
Then an estimator for \(\tau(x)\) is given by
\begin{align} \label{eq:theta} 
\hat{\tau}(x) = \left( \sum_{i=1}^n \alpha_i(x) (D_i - D_\alpha)(D_i - D_\alpha)^\top \right)^{-1} 
\sum_{i=1}^n \alpha_i(x) (D_i - D_\alpha)(Y_i - Y_\alpha),
\end{align} 
where \(D_\alpha = \sum \alpha_i(x) D_i\) and \(Y_\alpha = \sum \alpha_i(x) Y_i\). 
In our setting the treatment is binary and scalar, \(D_i\in\{0,1\}\). Hence, the local regression estimator simplifies to
\[
\hat{\tau}(x)
=
\frac{
\sum_{i=1}^{n}
\alpha_i(x)(D_i-D_{\alpha})(Y_i-Y_{\alpha})
}{
\sum_{i=1}^{n}
\alpha_i(x)(D_i-D_{\alpha})^2
}.
\]

Thus, the weights $\alpha_i(x)$ strongly influence whether we obtain a good estimator of $\tau(x)$. To construct favorable splits to produce the weights, the idea is that instead of using an exact loss criterion, which can be quite computationally expensive, the causal forest uses a gradient-based loss criterion to approximate the splitting rule used in the original causal trees, which is summarized in Algorithm \ref{algo:grf_ct} and \ref{algo:gradienttree} in Appendix \ref{app:algorithm}.  

To obtain a kernel-weighted version of the causal forest, we use kernel weights as sample weights. While the original framework does not formally introduce sample weights as a theoretical concept, they are supported in the \texttt{R} implementation and online documentation \citep{grf_package}. In the default implementation, each observation receives equal weight. Additional details are provided in Appendix~\ref{app:algorithm}. We use kernel weights only to balance the covariates, thereby reducing potential bias resulting from covariate imbalance.  When kernel weights are incorporated into the splitting criterion, the sample-weighted version of the gradient-based splitting rule is then defined as

$$\tilde{S}_w(C_1, C_2) = \frac{\left(\sum_{i \in C_1} w_i \rho_i\right)^2}{\sum_{i \in C_1} w_i} 
+ 
\frac{\left(\sum_{i \in C_2} w_i \rho_i\right)^2}{\sum_{i \in C_2} w_i}, $$
where $w_i$ are the kernel weights and $\rho_i \in \mathbb{R}$ are pseudo-outcomes computed within each parent node, as shown in Algorithm~\ref{algo:gradienttree}.

\subsection{Weighted X-learner}
\label{sec:Weighted_XF}
Metalearners are a class of meta-learning algorithms that address the problem of estimating the CATE by decomposing it into one or more subproblems. These subproblems can be addressed using any supervised learning or regression methods, the so-called base learners. As noted by \citet{okasa2022metalearners}, base learners can be any black-box method, such as random forests \citep{athey2015machine}, Bayesian regression trees \citep{Hill, single_learner_bart}, or neural networks, as long as they are consistent estimators and satisfy a set of regularity conditions. We focus on the X-learner, introduced by \citet{kunzel_metalearners_2019}, as a novel approach for estimating CATE, with the random forest as a base learner. In Appendix \ref{app:metalearners}, we also briefly review the S-learner and T-learner, which serve as foundations for the X-learner. Algorithm~\ref{algo:xlearner} summarizes the three main steps of the X-learner. 

\FloatBarrier 
\vspace{0.5cm}
\begin{algorithm}[ht]
    \SetAlgoLined
    \KwIn{$(X_i, Y_i, D_i)$, where $i = 1,\dots, N$}

    \textbf{Step 1:} Use a supervised learning or regression algorithm to estimate
    \[
        \mu_1(x) := \mathbb{E}[Y_i(1) | X_i=x], \quad
        \mu_0(x) \coloneqq \mathbb{E}[Y_i(0) | X_i=x], 
    \]
    using any supervised learning or regression algorithm on the treated and control observations, respectively. Denote the estimated functions from the first stage by \(\hat{\mu}_1(x)\) and \(\hat{\mu}_0(x)\). \\[0.5cm]
    \textbf{Step 2:} Use $\hat{\mu}_1(x)$ and $\hat{\mu}_0(x)$ to compute the imputed treatment effects:
   \begin{align*}
        \tilde{Y}^{1}_i &\coloneqq Y_i - \hat{\mu}_0(X_i),\quad i \in \mathcal{S}_1,  \\
        \tilde{Y}^{0}_i &\coloneqq \hat{\mu}_1(X_i) -Y_i,  \quad i \in \mathcal{S}_0.
   \end{align*}
    Then, use the imputed treatment effects as response variables  to estimate
\[
\tau_1(x) \coloneqq \mathbb{E}[\tilde{Y}_i^{1} \mid X_i=x, D_i=1],
\]
\[
\tau_0(x) \coloneqq \mathbb{E}[\tilde{Y}_i^{0} \mid X_i=x, D_i=0].
\]
The corresponding estimates are denoted by $\hat{\tau}_1(x)$ and $\hat{\tau}_0(x)$.    
    \textbf{Step 3:} Define the CATE estimate as a weighted average of $\hat{\tau}_1(x)$ and $\hat{\tau}_0(x)$:
    \[
        \hat{\tau}(x) \coloneqq f(x) \hat{\tau}_0(x) + (1 - f(x)) \hat{\tau}_1(x),
    \]
    where $f(x) \in [0, 1]$ is a weight function.

    \caption{X-Learner Algorithm for Estimating CATE}
    \label{algo:xlearner}
\end{algorithm}
\FloatBarrier 
The first step is to estimate the response functions $\mu_1(x$) and $\mu_0(x)$, which are then are used to compute the imputed treatment effects $\tilde{Y}_i^{1}$ and $\tilde{Y}_i^{0}$. Next, the imputed treatment effects serve as new response variables: $\tilde{Y}^1_i$ is used in the treatment group to estimate $\hat{\tau}_1(x)$, and $\tilde{Y}_0(x)$ is used in the control group to estimate $\hat{\tau}_0(x)$. Finally the CATE estimator is a weighted average of  $\hat{\tau}_1(x)$, and $\hat{\tau}_0(x)$. 

\citet{kunzel_metalearners_2019} suggest using an estimate of the propensity score $p(x)$ as the weight function, or even using $f\equiv0$ or $f\equiv 1$ if the number of treated units is significantly smaller or larger than the number of control units, or vice versa. However, when Assumption \ref{overlap} is violated, the X-learner with propensity score-based weights may produce extreme values. These extreme weights can lead to an over-reliance on either  $\hat{\tau}_1(x)$ or $\hat{\tau}_0(x)$, skewing the final CATE estimate toward one group’s imputed treatment effects.

Therefore, we propose the use of Kbal weights as an alternative to the propensity score weight function for the X-learner. The idea is to replace the last step of the X-learner algorithm with a new weight function $f_{kbal}$.  For the construction of $f_{kbal,i}$, 
we use $w^0$ and $w^1$ from Definition~\ref{defn:kernel-based weighting} and define
\begin{equation} \label{eq:function_kbal_weight}
    f_{kbal,i} = \frac{w_i^{0}}{w_i^{0}+w_i^{1} }\quad i=1,\dots N;
    \end{equation}
such that $f_{kbal,i} \in [0, 1]$ for all $i$.

In contrast to the standard X-learner, where the aggregation weight is written as a function \(f(x)\), the KBal-based aggregation weight is sample-specific because it does not have a closed-form representation as a function of covariates alone and is constructed from sample-level balancing weights; see Appendix~\ref{Appendix_Implementation} for details. Therefore, for each \(i\), we write the modified aggregation step as
\begin{align}
\label{eq:final_kbal_tau}
\hat{\tau}(X_i)
&=
f_{kbal,i}\hat{\tau}_0(X_i)
+
(1-f_{kbal,i})\hat{\tau}_1(X_i).
\end{align}

This sample-specific nature is common to optimization-based balancing weights. As emphasized by \citet{Zubizarreta2020}, such weights are tied to the empirical sample on which the balancing problem is solved. Hence, when new observations are added, the weights must either be recomputed on the enlarged sample or assigned through an additional rule, such as matching to similar observations or estimating a mapping from covariates to weights.

\begin{proposition}
\label{fkbal_conistency}
Let \(\hat{\tau}_0(X_i)\) and \(\hat{\tau}_1(X_i)\) be the second-stage estimators obtained from Algorithm~\ref{algo:xlearner}, and let \(f_{kbal,i}\in[0,1]\) denote the sample-specific aggregation weight constructed in \ref{eq:function_kbal_weight}. If \(\hat{\tau}_0(X_i)\) and \(\hat{\tau}_1(X_i)\) are consistent for \(\tau(X_i)\), then the aggregated estimator in \eqref{eq:final_kbal_tau} is also consistent for \(\tau(X_i)\).
\end{proposition}

\begin{proof}
See Appendix \ref{Appendix_Proofs}.
\end{proof}

To provide intuition, \eqref{eq:final_kbal_tau} illustrates how KBal weights determine the relative influence of $\hat{\tau}_0(X_i)$ and $\hat{\tau}_1(X_i)$. For each observation $i$, $f_{kbal,i}$ is obtained by normalizing $w_i^0$ and $w_i^1$ by their observation-specific sum. If an observation in the control group requires more reweighting for balance, this results a larger aggregation weight $f_{kbal,i}$, increasing reliance on $\hat{\tau}_0(X_i)$. Conversely, larger treatment-side balancing weights produce smaller $f_{kbal,i}$ values, increasing reliance on $\hat{\tau}_1(X_i)$. Note that, $f_{kbal,i}$ should not itself be interpreted as a balancing weight on the raw covariates. The original weights \(w^0\) and \(w^1\) are chosen to target covariate balance, whereas \(f_{kbal,i}\) is a transformed aggregation weight used only to combine the two X-learner components. Consequently, the balancing properties of \(w^0\) and \(w^1\) do not directly translate into a covariate-balance condition involving \(f_{kbal,i}\).

\section{Monte Carlo Simulations }
\label{ch:KBAL_sim_study}
To study finite-sample properties of optimization-based weighting for CATE estimation, we consider three simulation designs adapted from \citet{wager_estimation_2018}, \citet{Damour_Franks_2012}, \citet{kunzel_metalearners_2019} and \citet{rcate_Li}.
Each design is defined by the propensity score function $ p(x) $ and the response functions $ \mu_0(x) $ and $ \mu_1(x) $. Throughout, unconfoundedness holds and the error terms are homoskedastic. For observation $i$, potential outcomes are generated as
\begin{align} 
    Y_i(1) = \mu_1(X_i) + \varepsilon_i(1), \\
    Y_i(0) = \mu_0(X_i) + \varepsilon_i(0),
\end{align}
with $\varepsilon_i(1),\varepsilon_i(0) \overset{iid}{\sim} \mathcal{N}(0, 1)$.  
Then, the treatment assignment and the observed outcome are set as follows
\begin{equation}
        D_i \sim \text{Bern}(p(X_i)), \quad Y_i ^{obs} = Y_i(D_i), 
\end{equation} 
yielding samples $ (X_i, D_i, Y_i ^{obs}) $. 

We consider nine different simulation settings obtained by combining three different CATE structures, one linear and  two nonlinear, with three overlap conditions (good, moderate and poor).  
Overlap is varied by the joint distribution of covariates and the propensity score specification. Specifically, in the good- and poor-overlap settings, the covariates follow a normal distribution, $X_{iK} \sim \mathcal{N}(0, 1),$ while in the moderate overlap setting, they are standard uniformly distributed, $X_i \sim \mathcal U([0,1]^K)$. For simplicity, the propensity score functions just depend on the first column vector of $X$, denoted as $x_1$.  
In the good overlap setting, the propensity score is defined as 
\begin{equation}\label{eq:full_overlap}
    p(X_i=x) = \frac{1}{4} [1+\beta_{2,4}(x_1)],
\end{equation}
where $ \beta_{2,4} $ is the cumulative distribution function of the beta distribution with shape parameters $(2, 4)$. To obtain moderate or poor overlap the propensity score is instead defined as
\begin{equation} \label{eq:overlap_violation}
    p(X_i=x) = \frac{1}{1+\exp(-1-2x_1)}.
\end{equation}
Under Equation \eqref{eq:overlap_violation}, a normal covariate distribution produces poor overlap, while a uniform covariate distribution yields moderate overlap. The
resulting propensity score distributions are shown in
Figures~\ref{fig:propensities_S1}-\ref{fig:propensities_S3}
later in this section. Note that in the poor overlap setting the treatment and control group are of equal size, whereas in the moderate setting the samples are unbalanced, with the treatment group being considerably larger than the control group.   This design allows us to investigate the performance of kernel balancing weights under varying degrees of overlap violation.

We generate $ MC = 1000 $ independent samples, where each sample consists of $ N^{\text{train}} = 1000 $ training observations and $ N^{\text{test}} = 1000 $ test observations. The precision of CATE estimation  is evaluated using Root Mean Squared Error (RMSE), Mean Absolute Prediction Error (MAE) and Mean Absolute Bias, which are defined in Appendix~\ref{app:Appendix_monte}.

\subsection{Simulation S1 (Linear CATE)}
For the first data generating process, we adopt a linear response structure, based on the complex linear model of
\citet{kunzel_metalearners_2019} with modified coefficient vectors $ \beta_0 $ and $ \beta_1 $ and varying propensity score. 
\begin{align}
    \mu_0(x) &= x^\top \beta_0, \quad \text{with } \beta_0 \sim \text{Unif}([1,10]^{10}), \\
    \mu_1(x) &= x^\top \beta_1, \quad \text{with } \beta_1 \sim \text{Unif}([1,10]^{10}),
\end{align}
so that $\tau(x) = x^\top(\beta_1-\beta_0)$ is the treatment effect. Figure~\ref{fig:propensities_S1} reports the propensity score distributions under the three overlap settings.
\begin{figure}[H]
\caption{Propensity score distributions for Simulation S1 under good, moderate, and poor overlap}
    \includegraphics[width=1\textwidth]{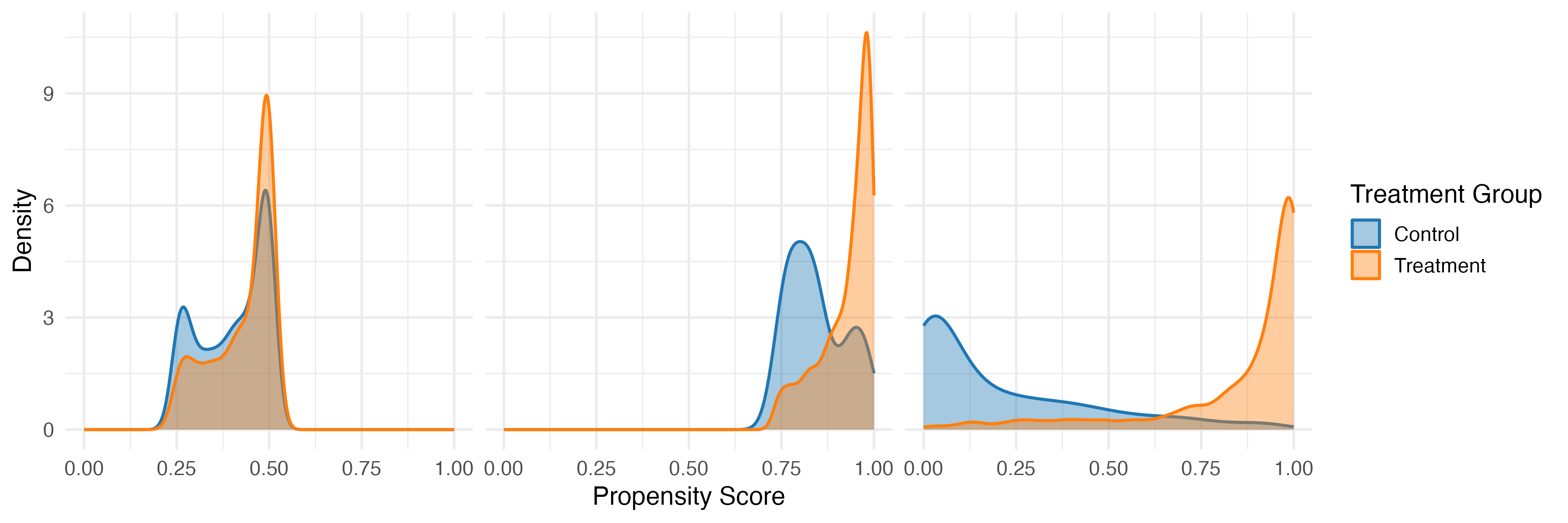} 
\caption*{\footnotesize\textit{Notes:} The panels show the distribution of the true propensity scores induced by the data-generating process. The left panel shows good overlap, the middle panel moderate overlap, and the right panel poor overlap.}
  \label{fig:propensities_S1}
\end{figure}
\subsection{Simulation S2 (Nonlinear CATE)}
This data setting implements a nonlinear treatment effect, which is based on \citet{lei_conformal_2021} and \citet{wager_estimation_2018}.
\begin{align*}
    \mu_0(x) &= 0,\\
    \mu_1(x) &=  f(x_1) \cdot f(x_2), \\\text{ with } f(x_j) &= \frac{2}{1+ \exp \{ -12(x_j-0.5) \} }, \quad j \in \{1,2\}.
\end{align*}
Again,  $x_1, x_2$ are the first and second column vector of $X$. Since $Y (0) \equiv 0$, this simplifies the setting and the true treatment effect is then $ \tau(x)=  \mu_1(X_i=x)$. 
The resulting propensity score distributions for this setting are illustrated in Figure~\ref{fig:propensities_S2}.
\begin{figure}[H]
\caption{Propensity score distributions for Simulation S2 under good, moderate, and poor overlap}    \includegraphics[width=1\textwidth]{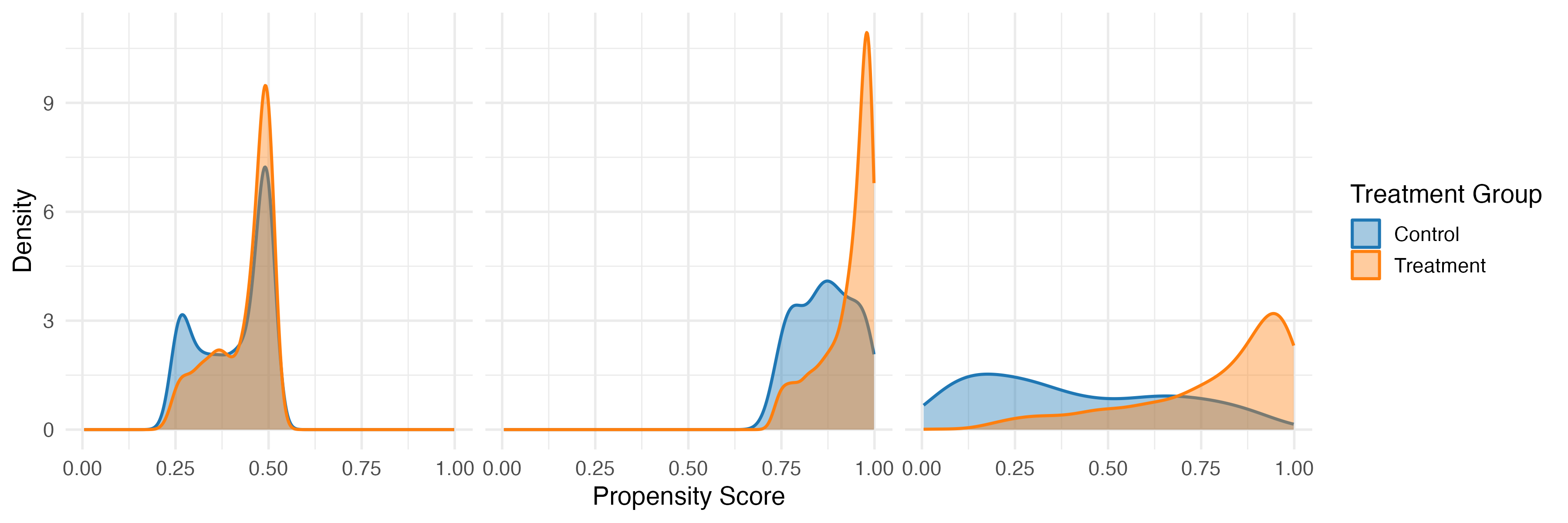}
\caption*{\footnotesize\textit{Notes:} The panels show the distribution of the true propensity scores induced by the data-generating process. The left panel shows good overlap, the middle panel moderate overlap, and the right panel poor overlap.}
  \label{fig:propensities_S2}
\end{figure}

\subsection{Simulation S3 (Complex nonlinear CATE)}
The third data setting features a treatment effect function that is nonlinear, which includes nonlinear covariate interaction terms. This design is based on \citet{rcate_Li}. While their original design includes contamination and outliers, we omit this aspect since it is not the focus of our study. The propensity score distributions for this setting are illustrated in Figure~\ref{fig:propensities_S3}.
\begin{align*}
\mu_1(x) &= \mu_0(x) + \tau(x), \quad \text{where} \\
\tau(x) &= 6\sin(2x_1)
  + 3(x_2+3)x_3
  + 9\tanh(0.5x_4)
  + 3x_5 \bigl(2\cdot\mathds{1}\{x_4>0\}-1\bigr) \\
&\quad + 3x_6 + 2x_7 + x_8 - 2x_9 - 4x_{10}, \\
\mu_0(x) &= 100 + 4x_1 + x_2 - 3x_3 - \tfrac{1}{2}\tau(x).
\end{align*}
Note that the appearance of $\tau(x)$ in the definition of $\mu_0(x)$ is part of the data-generating process. It does not change the treatment effect, since by definition
$$\mu_1(x)=\mu_0(x)+\tau(x),$$
and therefore
$$\mu_1(x)-\mu_0(x)=\tau(x).$$
\begin{figure}[H]
\caption{Propensity score distributions for Simulation S3 under good, moderate, and poor overlap}
    \includegraphics[width=1\textwidth]{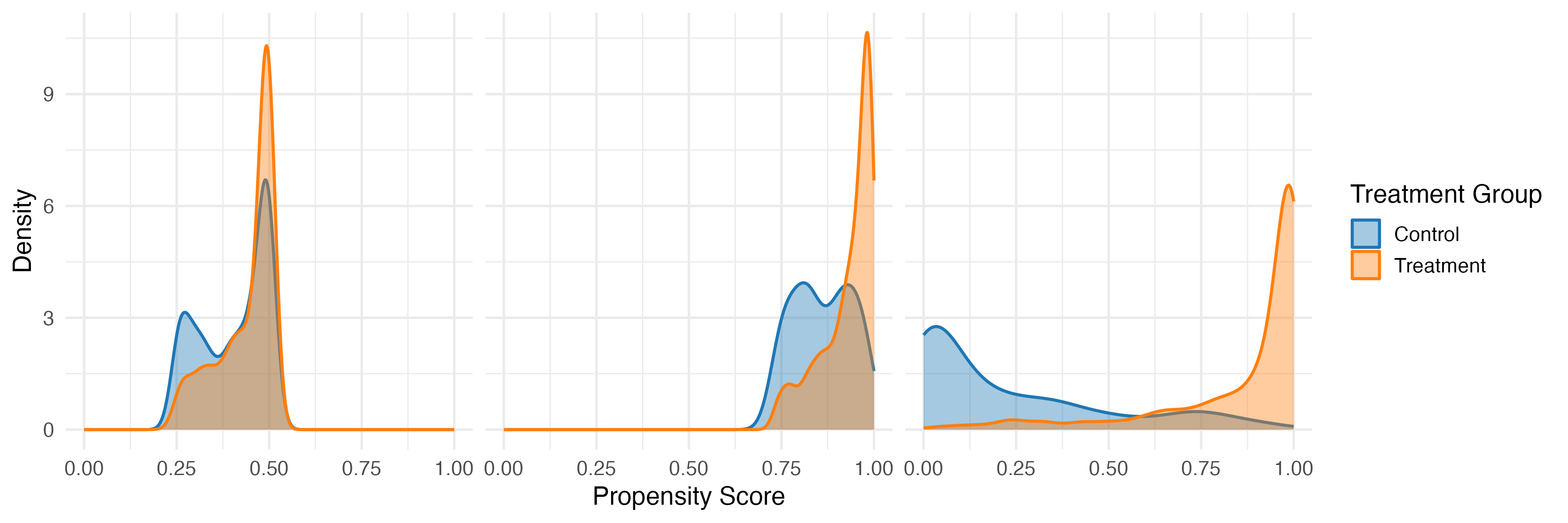} 
\caption*{\footnotesize\textit{Notes:} The panels show the distribution of the true propensity scores induced by the data-generating process. The left panel shows good overlap, the middle panel moderate overlap, and the right panel poor overlap.}
  \label{fig:propensities_S3}
\end{figure}
\subsection{Results}
\label{sec:results}
Table \ref{tab:S1_S2_S3_combined} reports our findings for the causal forest (CF), causal forest with trimmed inverse probability weighting (CF-IPW), X-learner with random forest base learners (XRF), and their kernel balancing variants (CF-KBal, XRF-KBal). The unweighted variants serve as benchmarks. We also report coverage and average interval width for the 95\% confidence intervals for completeness in Appendix~\ref{app:Appendix_monte}.

Under good overlap, estimation performance is comparable across methods in all three simulation settings. Differences in RMSE and MAE between weighted and unweighted estimators are small, and no method consistently outperforms the others. Bias differences are more noticeable for causal forests. In particular, for S1, where the CATE is linear, CF-IPW, and CF-KBal have slightly larger negative bias than CF. In S3 the pattern reverses, relative to CF, CF-IPW and CF-KBal reduces bias substantially by approximately 92\% and 90\%, respectively. As expected, when treated and control units are already sufficiently comparable over the covariate space, weighting only makes small adjustments and therefore has limited effect overall. As overlap deteriorates, the role of weighting becomes more important, but its effectiveness depends on the learner, the extent of the overlap violation, and the underlying CATE structure.

When overlap is moderate, performance differences become more pronounced and design-dependent. In S1, weighting does not change the benchmark performance significantly. Relative to CF, CF-IPW and CF-KBal slightly worsen performance, increasing RMSE by about 2.5\% and 1.5\%, from 10.611 to 10.874 and 10.768, respectively. 
Average bias also becomes slightly more negative, changing from -9.770 for CF to -10.055 for CF-IPW and -9.931 for CF-KBal. For the X-learner, XRF-KBal reduces RMSE by about 2.5\%, from 10.291 to 10.028, while the bias remains large, changing only slightly from -9.313 to -9.243. In contrast, in S2 and S3, where the treatment effect is nonlinear, kernel balancing improves the performance more clearly. In S2, CF-KBal reduces RMSE by around 16\%, from 1.307 to 1.093, while average bias decreases from -0.547 to -0.001. XRF-KBal also achieves large performance improvements, reducing RMSE by around 43\% and nearly halving bias. In S3, kernel balancing again substantially reduces bias, although RMSE improvements are more modest than in S2. Relative to CF, CF-KBal reduces RMSE from 3.243 to 3.126, corresponding to a decrease of around 4\%,  while average bias falls sharply from 1.429 to 0.040. For the X-learner, XRF-KBal reduces RMSE by 16\% and reduces the magnitude of the average bias by approximately 52\% relative to XRF.

Under poor overlap, kernel balancing continues to reduce bias in nonlinear settings. In S2, CF-KBal outperforms CF in terms of both RMSE and bias, with RMSE decreasing by about 16\% and the magnitude of the average bias reduced by roughly 74\% (from -0.240 to -0.063). Similarly, XRF-KBal reduces RMSE by 25\% and bias by 37\% (from -0.072 to -0.046). In S3, CF-KBal again yields a considerable bias reduction (from -1.607 to -0.038) with little change in RMSE, while XRF-KBal improves both RMSE and bias relative to XRF. In S1, the benefits of kernel balancing remain limited with no significant improvements. 

Furthermore, comparing CF and XRF, the results show that neither estimator class dominates uniformly. CF is not a metalearner, but \citet{kunzel_metalearners_2019} argue that it is similar to TRF and SRF (T-learner and S-learner with random forest as a base learner).  In the unweighted case, XRF often performs better than CF in terms of RMSE and MAE under moderate overlap, particularly in S1 and S3, while differences are smaller under good overlap. This is not surprising, as the X-learner is especially useful when the treatment and control groups differ in size, as they do in the moderate-overlap setting.  In the kernel-weighted case, XRF-KBal also tends to achieve lower RMSE and MAE than CF-KBal in several moderate- and poor-overlap settings, especially for nonlinear CATE designs. However, the ranking depends on the evaluation metric: CF-KBal often performs well in terms of bias reduction, even when XRF-KBal achieves lower prediction error. This is consistent with the broader conclusion of \citet{kunzel_metalearners_2019} that no single metalearner performs best across all settings.

Overall, kernel balancing affects CATE estimation through different mechanisms across methods. For the metalearners, performance improvements are primarily associated with bias reduction arising from propensity scores misspecification, which is consistent with the observation that under good overlap XRF and XRF-KBal perform similarly. In contrast, causal forests, which do not rely on propensity score weighting, benefit from kernel balancing through improved covariate balance in feature spaces that capture outcome-relevant nonlinearities and interactions, thereby attenuating residual imbalances even when overlap is good. 

\begin{table}[ht]
\centering
\caption{Monte Carlo results for CF and XRF estimators under varying overlap and treatment-effect complexity}
\label{tab:S1_S2_S3_combined}
\footnotesize
\setlength{\tabcolsep}{2.5pt}
\resizebox{\textwidth}{!}{%
\begin{tabular}{l l  r r r  r r r  r r r}
\\[-1.8ex]\hline
\hline \\[-1.8ex]
 &  & \multicolumn{9}{c}{Simulation Setting} \\
Overlap & Method
& \multicolumn{3}{c}{S1}
& \multicolumn{3}{c}{S2}
& \multicolumn{3}{c}{S3} \\
\cmidrule(lr){3-5} \cmidrule(lr){6-8} \cmidrule(lr){9-11}
 & 
& RMSE & MAE & Bias
& RMSE & MAE & Bias
& RMSE & MAE & Bias \\
\hline
\multirow{5}{*}{good}
& CF       
& \textbf{3.930} & \textbf{3.163} & \textbf{-1.015}
& 0.273 & 0.223 & 0.016
& 2.209 & 1.758 & 0.122 \\
& CF-IPW   
& 3.953 & 3.182 & -1.105
& \textbf{0.270} & \textbf{0.219} & 0.004
& \textbf{2.178} & \textbf{1.736} & \textbf{0.009} \\
& CF-KBal  
& 3.954 & 3.182 & -1.107
& 0.275 & 0.222 & \textbf{-0.003}
& 2.190 & 1.745 & 0.012 \\
& XRF      
& 3.941 & 3.173 & -0.917
& \textbf{0.273} & \textbf{0.219} & \textbf{0.023}
& 2.150 & 1.716 & 0.082 \\
& XRF-KBal 
& \textbf{3.934} & \textbf{3.168} & \textbf{-0.910}
& 0.275 & 0.221 & 0.023
& \textbf{2.143} & \textbf{1.711} & \textbf{0.073} \\
\midrule
\multirow{5}{*}{moderate}
& CF       
& \textbf{10.611} & \textbf{9.787} & \textbf{-9.770}
& 1.307 & 0.828 & -0.547
& 3.243 & 2.627 & -1.429 \\
& CF-IPW   
& 10.874 & 10.070 & -10.055
& 1.160 & \textbf{0.785} & -0.361
& 3.138 & 2.524 & -0.120 \\
& CF-KBal  
& 10.768 & 9.948 & -9.931
& \textbf{1.093} & 0.887 & \textbf{-0.001}
& \textbf{3.126} & \textbf{2.513} & \textbf{0.040} \\
& XRF      
& 10.291 & 9.353 & -9.313
& 1.062 & 0.734 & -0.224
& 2.980 & 2.392 & -0.985 \\
& XRF-KBal 
& \textbf{10.028} & \textbf{9.259} & \textbf{-9.243}
& \textbf{0.610} & \textbf{0.435} & \textbf{-0.118}
& \textbf{2.501} & \textbf{2.003} & \textbf{-0.450} \\
\midrule
\multirow{5}{*}{poor}
& CF       
& 10.274 & 8.198 & 0.870
& 1.026 & \textbf{0.444} & -0.240
& \textbf{8.399} & 6.616 & -1.607 \\
& CF-IPW   
& 11.030 & 8.846 & \textbf{-0.024}
& 0.881 & 0.450 & -0.100
& 8.731 & 6.865 & 0.115 \\
& CF-KBal  
& \textbf{10.261} & \textbf{8.190} & -0.142
& \textbf{0.863} & 0.460 & \textbf{-0.063}
& 8.428 & \textbf{6.612} & \textbf{-0.038} \\
& XRF      
& \textbf{12.214} & \textbf{9.767} & \textbf{6.027}
& 0.789 & 0.406 & -0.072
& 8.617 & 6.809 & -0.810 \\
& XRF-KBal 
& 12.302 & 9.857 & 6.374
& \textbf{0.593} & \textbf{0.319} & \textbf{-0.046}
& \textbf{8.242} & \textbf{6.336} & \textbf{0.474} \\
\hline
\end{tabular}%
}
\begin{flushleft}
\footnotesize
\textit{Notes:} Results are averaged over 1000 Monte Carlo replications. Best values are shown in bold relative to the unweighted benchmark within each method class.
\end{flushleft}
\end{table}


\section{Semi-Synthetic Application}
\label{ch:emp_app}

Evaluating causal inference models is challenging because counterfactual outcomes are unobservable in practice. This is often referred to as the  “fundamental problem of causal inference” \citep{holland_statistics_1986}. Although Assumptions \ref{SUTVA}-\ref{Linearity of Expected Outcome} can justify the estimation of causal effects, they do not provide observable ground-truth treatment effects at the unit level.
Consequently, observational datasets hardly ever provide a ground truth against which the performance of causal inference methods can be measured. Typically, there are two ways to address this issue. One is using real-world data from randomized controlled
trials (RCTs). However, the problem in using data from RCTs for our setting is that the treatment is received randomly. As a result, there is no imbalance between the treated and control distributions and no overlap violations, making weighting methods mostly redundant.
The other approach is to use semi-synthetic datasets. These use real covariates with simulated outcomes and have become a widely used tool\footnote{There is, however, an ongoing debate about the use and limitation of semi-synthetic datasets, which is beyond the scope of this work. For details, see \citet{curth2021doinggreatestimatingcate,poinsot2025positioncausalmachinelearning}.}  in causal machine learning \citep{Hill,dorie2016npci,dorie_2019_do_it_yourself, Wendling_2018, Knaus_2020}. 

\subsection{Data Preprocessing and Setting}
To evaluate the performance of our kernel balancing methods, we use data from the Infant Health and Development Program (IHDP) dataset \citep{IHDP}, which was initially designed based on a randomized controlled trial that started in 1985 to investigate the impact of home visits from doctors on cognitive outcomes in low-birth-weight, premature infants. The data set includes 25 variables with pre-treatment characteristics, such as child-specific metrics (e.g., birth weight, head circumference, weeks born preterm), pregnancy behaviors (e.g., maternal smoking, maternal alcohol use), and maternal demographic factors (e.g., age, marital status, education level, employment during pregnancy). \citet{Hill} used this data to construct a semi-synthetic version, where the outcomes are generated using a (non-)linear response surface while using the original covariates. The transformed IHDP dataset has become a benchmark for evaluating new causal machine-learning methods \citep{Johansson_Dragonet,alaa2017bayesianinference,alaa18a_limits,API_Lin_2019,cheng2022_counterfactual} because it allows us to evaluate the performance of our estimators based on the true treatment effects. Furthermore, to simulate observational data, \citet{Hill} introduced selection bias by removing a non-random subset of the treatment group (children of non-white mothers).  This means treatment assignment is no longer random in the transformed dataset, as maternal race becomes related to whether a unit remains in the treatment group. This adjustment created an unbalanced dataset with 139 treated and 608 control units. Since these covariates are also used to generate the simulated potential outcomes and are therefore predictive of the outcome, a simple difference-in-means estimator would be biased. 

For our analysis, we use two variants of the transformed IHDP dataset. The first is the simulated IHDP dataset provided by \citet{Johansson_Dragonet}, who made the dataset publicly available at \url{https://www.fredjo.com/}. This dataset consists of 1000 repetitions generated based on setting A from the NPCI \texttt{R} package \citep{dorie2016npci}. The outcomes are simulated according to
\[
Y(0) \sim \mathcal{N}\left(\exp((X + A) \beta_A), 1\right)
 \quad \text{and} \quad Y(1) \sim \mathcal{N}\left(X\beta_A - \omega_A, 1\right),
\]
where \( X \) is a matrix of standardized covariate values, with the first column being a vector of ones,  \( A \) is a matrix of the same dimensions as \( X \), with all elements set to 0.5. The coefficients in \( \beta_A \) are randomly drawn from the set \(\{0, 0.1, 0.2, 0.3, 0.4\}\) with probabilities \(\{0.6, 0.1, 0.1, 0.1, 0.1\}\). Finally, \( \omega_A \) is a shift parameter and used to normalize the average treatment effect in the target population.  In our analysis, $\omega_A$ is already fixed in the simulated outcomes and is not modified.

For the second variant, we use setting B from the NPCI \texttt{R} package \citep{dorie2016npci}, where the outcomes are simulated according to 
\[
Y(0) \sim \mathcal{N}(X^*\beta_B, 1) \quad \text{and} \quad Y(1) \sim \mathcal{N}(X^*\beta_B+4, 1),
\]

where $X^*$ is a matrix of standardized covariate values that also includes quadratic and interaction terms. Furthermore, $\beta_B$ is defined as \( \beta_B= (\beta_1,\beta_2) \), where $ \beta_1 $ is randomly drawn from the set $\{0, 1, 2\}$ with probabilities $\{0.6, 0.3, 0.1\}$, and $ \beta_2 $ is randomly drawn from the set $\{0, 0.5, 1\}$ with probabilities $\{0.8, 0.15, 0.05\}$. The resulting coefficient vector is then applied to all observations in that repetition.

While setting A uses a nonlinear response surface, setting B  remains linear in its parameters. We use a $50\% / 50 \%$ train-test split for our analysis and the overlap assumption is moderate, as indicated in Figure \ref{fig:propensities_IHDP} in Appendix \ref{Appendix_figures}, resulting in a similar setting as in our simulation study in Section \ref{ch:KBAL_sim_study}. We average over 1000 repetitions and report the aggregated statistics, which we have already used in Section \ref{sec:results}. 

\subsection{Results}
\label{ch:emp_app_results}
Table \ref{tab:results_ihdp_combined} presents the results for Settings A and B. As before, the unweighted estimators serve as benchmarks. In Setting A, kernel balancing generally improves estimation accuracy, but the improvement varies across methods. For causal forests, CF-KBal achieves small reductions in MAE and bias, while RMSE remains close to that of the CF estimator. The effects are more pronounced for the X-learner. XRF-KBal exhibits noticeably lower RMSE (about 10\%) and MAE (about 12\%), and the average bias is close to zero compared to XRF.

In Setting B, the differences among estimators are less pronounced. For causal forests, RMSE and MAE are comparable across CF, CF-IPW, and CF-KBal, with CF exhibiting the smallest bias. Similarly, XRF-KBal achieves only minor reductions in RMSE and MAE relative to XRF, and the bias remains essentially unchanged. These results suggest that kernel balancing has a limited effect in Setting B, despite having moderate overlap by construction. The observed performance patterns indicate that overlap is sufficient, as evidenced by the absence of systematic bias reduction in the metalearner and the slight increase in bias for the causal forest variants. This behavior mirrors the pattern observed in Setting S1 of the simulation study under good overlap.

Comparing CF and XRF, the IHDP results again show that neither estimator class dominates uniformly across all criteria, but the XRF and its variant perform better according to more criteria. 
In both settings, XRF achieves lower RMSE and MAE than CF, which is again consistent with the fact that the X-learner is designed to perform well when the treatment and control groups differ in size. However, the bias comparison is more mixed: in setting A, XRF-KBal has the smallest bias overall, whereas in setting B the unweighted CF has the smallest bias. 
 
Overall, the IHDP results mirror the simulation findings, indicating that kernel balancing primarily improves performance when residual imbalance or propensity-score misspecification affects estimation, but offers limited gains when effective overlap is sufficient.

\FloatBarrier
\begin{table}[ht]
\centering
\caption{Semi-synthetic IHDP results for Settings A and B}
\label{tab:results_ihdp_combined}
\footnotesize
\setlength{\tabcolsep}{3pt}
\resizebox{\textwidth}{!}{
\begin{tabular}{l l r r r r r r r r r r}
\\[-1.8ex]\hline
\hline \\[-1.8ex]
 &  & \multicolumn{10}{c}{IHDP Setting} \\
Method 
& 
& \multicolumn{5}{c}{A}
& \multicolumn{5}{c}{B} \\
\cmidrule(lr){3-7} \cmidrule(lr){8-12}
 & 
& RMSE & MAE & Bias & Coverage & Width
& RMSE & MAE & Bias & Coverage & Width \\
\hline

CF 
& 
& 4.741 & 3.459 & -0.346 & 0.367 & 3.327
& 4.051 & 3.073 & \textbf{0.012} & 0.336 & 3.190 \\

CF-IPW
& 
& \textbf{4.668} & 3.349 & -0.127 & 0.398 & 3.336
& \textbf{4.038} & \textbf{3.072} & 0.073 & 0.382 & 3.686 \\

CF-KBal
& 
& 4.702 & \textbf{3.343} & \textbf{0.051} & \textbf{0.445} & 3.899
& 4.036 & 3.077 & 0.079 & \textbf{0.415} & 4.071 \\

XRF
& 
& 4.165 & 2.911 & -0.042 & 0.488 & 3.599
& 3.702 & 2.816 & \textbf{0.253} & 0.412 & 3.723 \\

XRF-KBal
& 
& \textbf{3.736} & \textbf{2.547} & \textbf{0.004} & \textbf{0.550} & 3.500
& \textbf{3.583} & \textbf{2.724} & 0.258 & 0.422 & 3.687 \\

\hline
\end{tabular}
}

\begin{notes}
 Results are averaged over 1000 IHDP replications. RMSE, MAE, and Bias are computed with respect to the true individual treatment effects. Coverage denotes empirical coverage of nominal 95\% intervals, and Width denotes average interval length. Bold values indicate the best performance relative to the unweighted benchmark within each method class.
\end{notes}
\end{table}
\FloatBarrier

\section{Conclusion and Outlook}
Conventional weighting approaches in observational studies typically rely on estimated propensity scores, either directly through inverse probability weighting, trimming, and matching, or indirectly in doubly robust estimators. While propensity-score methods are popular because they are simple and easy to implement, they can become unstable when overlap is limited. In such settings, extreme weights may arise and lead to unreliable treatment-effect estimates. Focusing on tree-based methods, we study the role of kernel balancing in CATE estimation and how its benefits vary with overlap conditions and treatment effect complexity.

For metalearners, we implemented kernel balancing through a modified aggregation weight in the X-learner, replacing the standard propensity-score-based aggregation function. For causal forests, we incorporated kernel balancing through sample weights that reweighted treatment and control units into a common pseudo-population. This adjusted the covariate distribution and addressing bias from covariate imbalance.

Our Monte Carlo results showed that kernel balancing improves estimation accuracy in settings where treatment effects are smooth and nonlinear, particularly under covariate imbalance or overlap violation. This finding is consistent with the argument of \citet{hazlett2016kernel}, who show that when outcomes depend on nonlinear functions of the covariates, conventional matching and weighting estimators can exhibit substantial bias even when raw covariates are well balanced. In contrast, when overlap is good, propensity-score-based methods already induced approximate balance in the outcome-relevant covariate functions, so that kernel balancing induces a similar pseudo-population and achieves comparable bias and estimation accuracy. 

Our empirical application to the Infant Health and Development Program data supports these findings. Performance gains from kernel balancing mirror the observed patterns in the simulation study, with clearer improvements in settings where residual imbalance or model misspecification plays a more prominent role, and more limited gains when overlap is effectively sufficient.

There are several directions for further research. While this study focused on kernel balancing as a special case of an optimization program, future work could examine either alternative optimization-based weighting methods or kernel balancing with different objective functions. In particular, in high-dimensional settings kernel balancing is computationally infeasible. To reduce runtime, this study relied on a linear kernel approximation, but it remains an open question how this approximation affects covariate balance, treatment-effect estimation, and the resulting performance gains. Future research could therefore investigate the trade-offs introduced by kernel approximations and explore tuning or approximation strategies that improve computational efficiency without sacrificing estimation accuracy.  For example, limiting the number of constraints for which exact balance is enforced has been suggested by \citet{cousineau_estimating_2022}. Related ideas are formally studied in the robust optimization literature, which provides guarantees for optimization under reduced or relaxed constraint sets \citep{Robust_Opimization_Minimal}.

Finally, it would be interesting to investigate how recently proposed outcome-guided kernel balancing methods, such as forest kernel balancing \citep{Forest_Kernel_Balancing_Weights}, could be integrated into the weighted X-learner and causal forest frameworks considered here, and to what extent they affect estimation accuracy and computational performance.

\acks{I would like to thank Christoph Hanck for valuable feedback on earlier drafts of this manuscript and I am grateful to seminar participants at the RuhrMetrics Research Seminar and the EuroCIM 2024 for helpful comments and discussions. I also acknowledge partial financial support from TRR 391 Spatio-temporal Statistics for the Transition of Energy and Transport (520388526) by the Deutsche Forschungsgemeinschaft (DFG, German Research Foundation) and from the Rhine-Ruhr Center for Scientific Data Literacy (DKZ.2R) by the German Federal Ministry of Education and Research (BMBF).  During the preparation of this manuscript, I used OpenAI’s large language models (up to and including GPT-5.5) to assist with proofreading and language editing. I reviewed and verified all generated content and am solely responsible for the accuracy of the final manuscript and any remaining errors.
}

\bibliography{overall_citation}
\appendix

\section{Additional Monte Carlo Results}   \label{app:Appendix_monte}
\subsection{Performance Criteria for Simulation Study}
Let \(M\) denote the number of Monte Carlo replications and recall that \(N_{\text{test}}\) is the number of test observations in each replication. The precision of CATE estimation is evaluated using root mean squared error (RMSE), mean absolute error (MAE), and average bias:
\begin{align*}
\overline{\mathrm{RMSE}}
&=
\frac{1}{M}
\sum_{mc=1}^{M}
\left[
\sqrt{
\frac{1}{N_{\text{test}}}
\sum_{i=1}^{N_{\text{test}}}
\left(
\tau(X_{i,mc})-\hat{\tau}(X_{i,mc})
\right)^2
}
\right],
\\
\overline{\mathrm{MAE}}
&=
\frac{1}{M}
\sum_{mc=1}^{M}
\left[
\frac{1}{N_{\text{test}}}
\sum_{i=1}^{N_{\text{test}}}
\left|
\tau(X_{i,mc})-\hat{\tau}(X_{i,mc})
\right|
\right],
\\
\overline{\mathrm{Bias}}
&=
\frac{1}{M}
\sum_{mc=1}^{M}
\left[
\frac{1}{N_{\text{test}}}
\sum_{i=1}^{N_{\text{test}}}
\left(
\hat{\tau}(X_{i,mc})-\tau(X_{i,mc})
\right)
\right].
\end{align*}

We additionally report pointwise confidence interval coverage and the width of the confidence interval. For observation $i$ in Monte Carlo replication $mc$, the nominal 95\% confidence interval is defined as
\[\mathrm{CI}_{i,mc}
=
\left[
\hat{\tau}(X_{i,mc}) - 1.96\,\widehat{\mathrm{se}}\{\hat{\tau}(X_{i,mc})\},
\;
\hat{\tau}(X_{i,mc}) + 1.96\,\widehat{\mathrm{se}}\{\hat{\tau}(X_{i,mc})\}
\right],\]
where the standard error is obtained from the built-in variance estimator in the \texttt{grf} package for causal forests and from the stratified nonparametric bootstrap for metalearners, as described in more detail below. Coverage and width are then defined as
\begin{align*}
\overline{\mathrm{Coverage}}
&=
\frac{1}{M}
\sum_{mc=1}^{M}
\left[
\frac{1}{N_{\text{test}}}
\sum_{i=1}^{N_{\text{test}}}
\mathds{1}
\left\{
\tau(X_{i,mc}) \in \mathrm{CI}_{i,mc}
\right\}
\right],
\\
\overline{\mathrm{Width}}
&=
\frac{1}{M}
\sum_{mc=1}^{M}
\left[
\frac{1}{N_{\text{test}}}
\sum_{i=1}^{N_{\text{test}}}
\left(
\mathrm{CI}^{\mathrm{upper}}_{i,mc}
-
\mathrm{CI}^{\mathrm{lower}}_{i,mc}
\right)
\right].
\end{align*}

\begin{table}[ht]
\centering
\caption{Coverage and interval width for CF and XRF estimators under varying overlap and treatment-effect complexity}
\label{tab:S1_S2_S3_coverage_width}
\footnotesize
\setlength{\tabcolsep}{2.5pt}
\begin{tabular}{l l  r r  r r  r r}
\\[-1.8ex]\hline
\hline \\[-1.8ex]
 &  & \multicolumn{6}{c}{Simulation Setting} \\
Overlap & Method
& \multicolumn{2}{c}{S1}
& \multicolumn{2}{c}{S2}
& \multicolumn{2}{c}{S3} \\
\cmidrule(lr){3-4} \cmidrule(lr){5-6} \cmidrule(lr){7-8}
 & 
& Coverage & Width
& Coverage & Width
& Coverage & Width \\
\hline

\multirow{5}{*}{good}
& CF       
& 0.219 & 2.254
& 0.612 & 0.545
& 0.387 & 2.226 \\
& CF-IPW   
& 0.220 & 2.271
& \textbf{0.721} & 0.653
& \textbf{0.390} & 2.214 \\
& CF-KBal  
& \textbf{0.220} & 2.277
& 0.619 & 0.541
& 0.390 & 2.225 \\
& XRF      
& \textbf{0.336} & 3.528
& 0.759 & 0.685
& \textbf{0.338} & 1.884 \\
& XRF-KBal 
& 0.311 & 3.246
& \textbf{0.778} & 0.710
& 0.326 & 1.813 \\
\midrule

\multirow{5}{*}{moderate}
& CF       
& 0.015 & 2.160
& 0.358 & 0.540
& 0.325 & 2.811 \\
& CF-IPW   
& 0.020 & 3.182
& \textbf{0.382} & 0.728
& 0.385 & 3.240 \\
& CF-KBal  
& \textbf{0.024} & 3.336
& 0.252 & 1.062
& \textbf{0.414} & 3.542 \\
& XRF      
& \textbf{0.033} & 3.026
& 0.520 & 1.064
& \textbf{0.395} & 3.125 \\
& XRF-KBal 
& 0.025 & 3.290
& \textbf{0.561} & 0.708
& 0.363 & 2.388 \\
\midrule

\multirow{5}{*}{poor}
& CF       
& 0.316 & 8.563
& 0.645 & 0.392
& 0.336 & 7.255 \\
& CF-IPW   
& \textbf{0.447} & 14.240
& \textbf{0.668} & 0.609
& \textbf{0.385} & 8.921 \\
& CF-KBal  
& 0.380 & 10.532
& 0.548 & 0.539
& 0.366 & 7.914 \\
& XRF      
& \textbf{0.335} & 10.879
& 0.783 & 0.781
& \textbf{0.287} & 6.304 \\
& XRF-KBal 
& 0.300 & 9.836
& \textbf{0.799} & 0.677
& 0.271 & 5.441 \\
\hline
\end{tabular}

\begin{notes}
 Results are averaged over 1000 Monte Carlo replications. Coverage denotes the share of 95\% confidence intervals that contain the true treatment effect. Width denotes the average interval length. Width should be interpreted jointly with coverage, since narrower intervals are only desirable when empirical coverage is adequate.
\end{notes}

\end{table}

\subsection{Coverage}
For causal forests, we rely on the built-in variance estimation procedure implemented in the \texttt{grf} package, which is based on the Bootstrap of Little Bags (BLB) method introduced by \citet{athey_generalized_2019}. The idea is to approximate the sampling variance of $\hat{\tau}(x)$ by exploiting the ensemble structure of the forest. The trees are grown on different subsamples of the data and grouped into smaller collections (“little bags”). Let $\hat{\tau}_{g,b}(x)$ denote the CATE prediction of
tree $b$ in group $g$, where each group contains $\ell$ trees. The group-level average is defined as 
\[
\bar{\tau}_g(x) = \frac{1}{\ell} \sum_{b=1}^{\ell} \hat{\tau}_{g,b}(x) \qquad g = 1,\dots,G.
\]
Trees within a group share the same subsample and are therefore dependent, while trees from different groups are approximately independent. The group-level mean predictions can then be written as
\begin{align*}
    \bar{\tau}(x) = \frac{1}{G} \sum_{g=1}^{G} \bar{\tau}_g(x).
\end{align*}
Under honesty, discussed in Section~\ref{sec:Weighted_CF}, and regularity conditions, variation across the group means
$\bar{\tau}_g(x)$ consistently estimates the sampling variability of the forest
estimator.  The resulting variance estimator can be written as
\[
\widehat{\Var}(\hat{\tau}(x)) \approx \widehat{V}(x)^{-2} \, \widehat{H}_n(x),
\]
where $\widehat{H}_n(x)$ captures the variability across little bags and
$\widehat{V}(x)$ is a local curvature term reflecting treatment variation.
The \texttt{grf} package \citep{grf_package} computes this quantity internally by setting \texttt{estimate.variance = TRUE} in prediction and provides pointwise confidence intervals based on a Gaussian approximation. Although these confidence intervals are constructed using the estimated standard errors defined above, the variance estimators implemented in \texttt{grf} primarily capture sampling variability and do not adjust for finite-sample bias \citep{wager_estimation_2018}. Coverage can therefore deteriorate when the estimator becomes bias-dominated rather than variance-dominated, for example in designs with strong nonlinear treatment effects, limited overlap, or model misspecification \citep{Comparing_Meta_Learners}.

For metalearners, we follow \citet{kunzel_metalearners_2019} and use stratified nonparametric bootstrap intervals with a normal approximation. In each bootstrap replication $r=1,\dots,R$ treatment and control observations are resampled with replacement according to the original group sizes, and the learner is refitted. We calculate pointwise standard errors as the empirical standard deviation across bootstrap replications, and form confidence intervals using a normal approximation (see Algorithm~ \ref{algo:bootstrap_meta} for details).
We use $R=50$ bootstrap replications. While this number is relatively small, we try to keep the runtime manageable, since inference is not the focus of this study. In additional checks with larger $R$, the resulting coverage patterns were similar. Table~\ref{tab:S1_S2_S3_combined} reports nominal 95\% pointwise confidence interval coverage for $\hat{\tau}(x)$. Coverage falls below nominal levels for both causal forests and metalearners in all our simulation settings, although there are differences between estimator classes.

Under good overlap, both CF-IPW and CF-KBal show only minor changes in coverage relative to CF in S1 and S3. In S2, CF-IPW produces the largest increase, raising coverage by approximately 17.8\% compared to CF. For the X-learner, coverage is generally higher than for causal forests, but the effect of weighting varies across settings. In S1 and S3, weighting slightly reduces coverage relative to XRF, by about 7.4\% and 3.6\%, respectively, whereas in S2, XRF-KBal yields a modest increase of approximately 2.5\%.

When overlap is moderate, all methods show very poor coverage in S1, and weighting does not substantially improve this outcome. Although interval widths increase for the weighted estimators, particularly for CF-IPW and CF-KBal, these wider intervals still fail to improve coverage. In S2 and S3, the pattern is more mixed. For causal forests, weighting generally increases interval width, but the effect varies by setting. In S2, CF-IPW increases coverage by about 7\%, whereas CF-KBal decreases coverage by approximately 30\% while nearly doubling the interval width. In S3, both weighted causal forest estimators improve coverage relative to CF, with gains of about 18\% for CF-IPW and 27\% for CF-KBal, and width increases of about 15\% and 26\%, respectively. For the X-learner, XRF-KBal performs particularly well in S2, increasing coverage by about 8\% while reducing interval width by roughly 33\%.

Under poor overlap, both weighted causal forest estimators improve coverage relative to CF in S1, but only at the expense of substantially wider intervals. CF-IPW increases coverage by about 41\% and width by roughly 66\%, while CF-KBal raises coverage by about 20\% with a smaller width increase of around 23\%. In S2, the pattern is again mixed: CF-IPW slightly improves coverage by about 4\%, but with a width increase of approximately 55\%, whereas CF-KBal reduces coverage by about 15\% while still increasing width by roughly 38\%. For the X-learner, XRF-KBal performs better, increasing coverage by about 2\% while reducing the interval width by approximately 13\%. In S3, both weighted causal forest estimators again improve coverage, but only with wider intervals. 

In general, providing formal confidence intervals is more challenging for metalearners than for causal forests, as these methods lack a unified asymptotic theory \citep{kunzel_metalearners_2019,Comparing_Meta_Learners}. Bootstrap-based confidence intervals are therefore common, but they typically fail to adjust for finite-sample bias in $\hat{\tau}(x)$ and tend to undercover. Moreover, in settings with limited or violated overlap, \citet{kunzel_metalearners_2019} document that confidence intervals may even become tighter and coverage may deteriorate, as tree-based learners extrapolate into regions of the covariate space with little or no support.

In our simulations, kernel balancing reduces average bias in settings S2 and S3 and, in several cases, also improves coverage relative to the corresponding unweighted estimators. However, coverage is still far below the nominal 95\% level. An intuition for this is that a small average bias does not rule out substantial bias at individual covariate values. Bias may cancel out on average, even if $\hat{\tau}(x)$ is systematically too high or too low, and the confidence intervals in that region will be centered incorrectly. Since coverage is evaluated pointwise, this kind of local misspecification can lead to undercoverage. The comparatively large MAE values are consistent with this explanation.
\subsection{Additional Figures and Tables}
\label{Appendix_figures}
\begin{figure}[!htbp]
\centering
\caption{Covariate Balance Before and After Weighting}
\label{fig:balance_3x3}
\begin{minipage}[t]{0.45\textwidth}
    \centering
    \includegraphics[width=\linewidth]{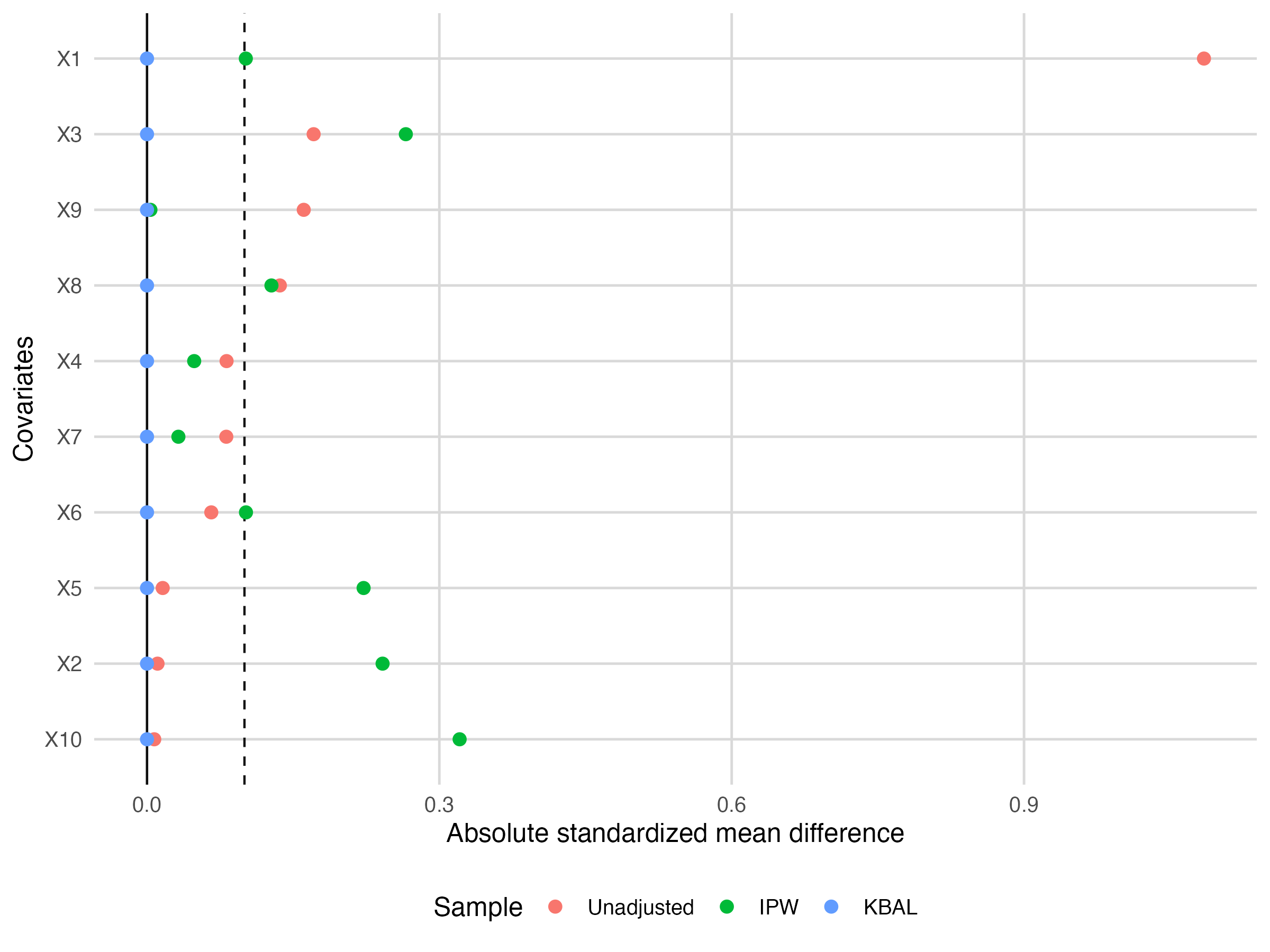}

    \smallskip
    (a) S1: Moderate overlap
\end{minipage}\hfill
\begin{minipage}[t]{0.45\textwidth}
    \centering
    \includegraphics[width=\linewidth]{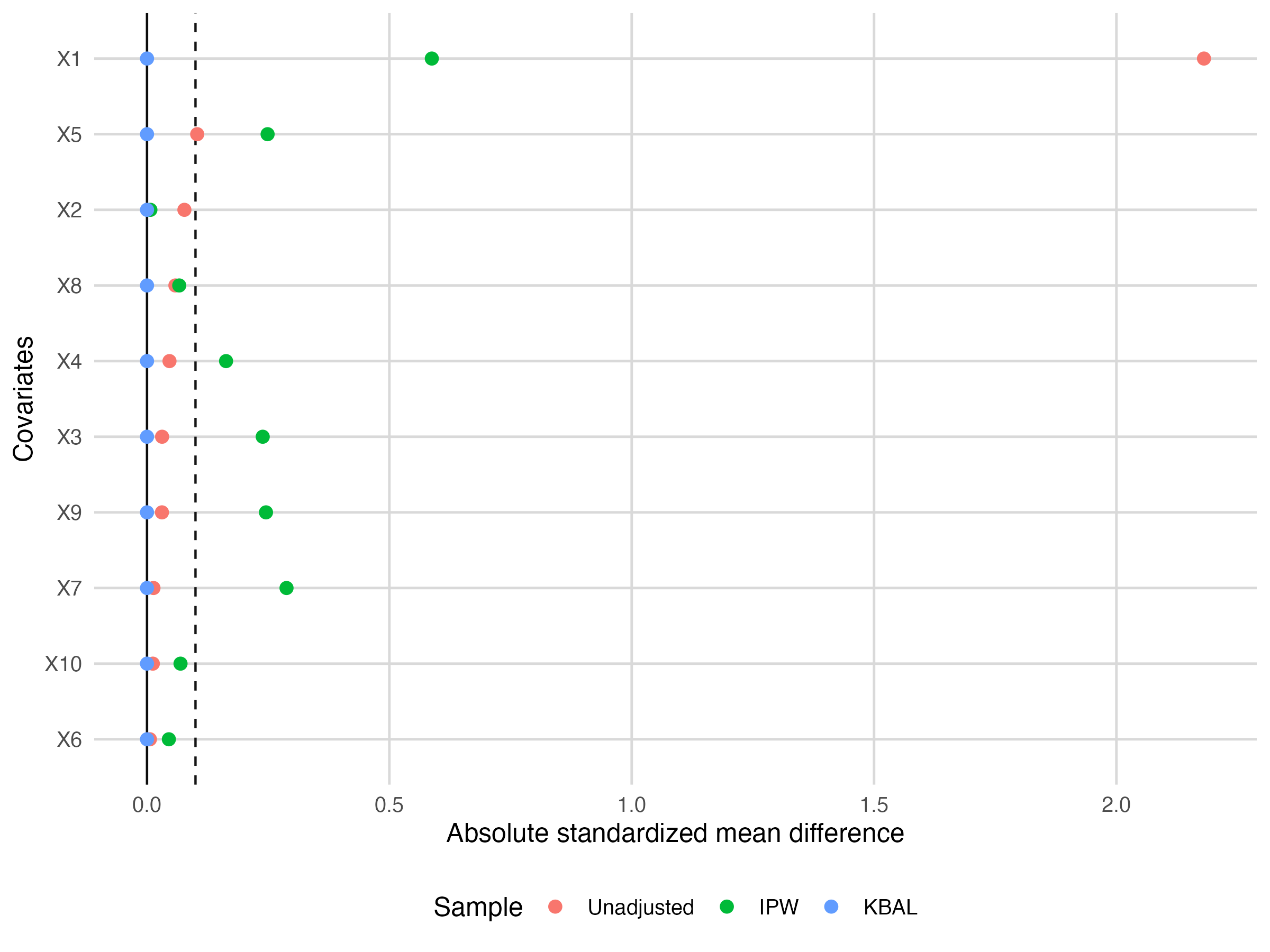}
    \smallskip
    (b) S1: Poor overlap
\end{minipage}
\vspace{0.8em}
\begin{minipage}[t]{0.45\textwidth}
    \centering
    \includegraphics[width=\linewidth]{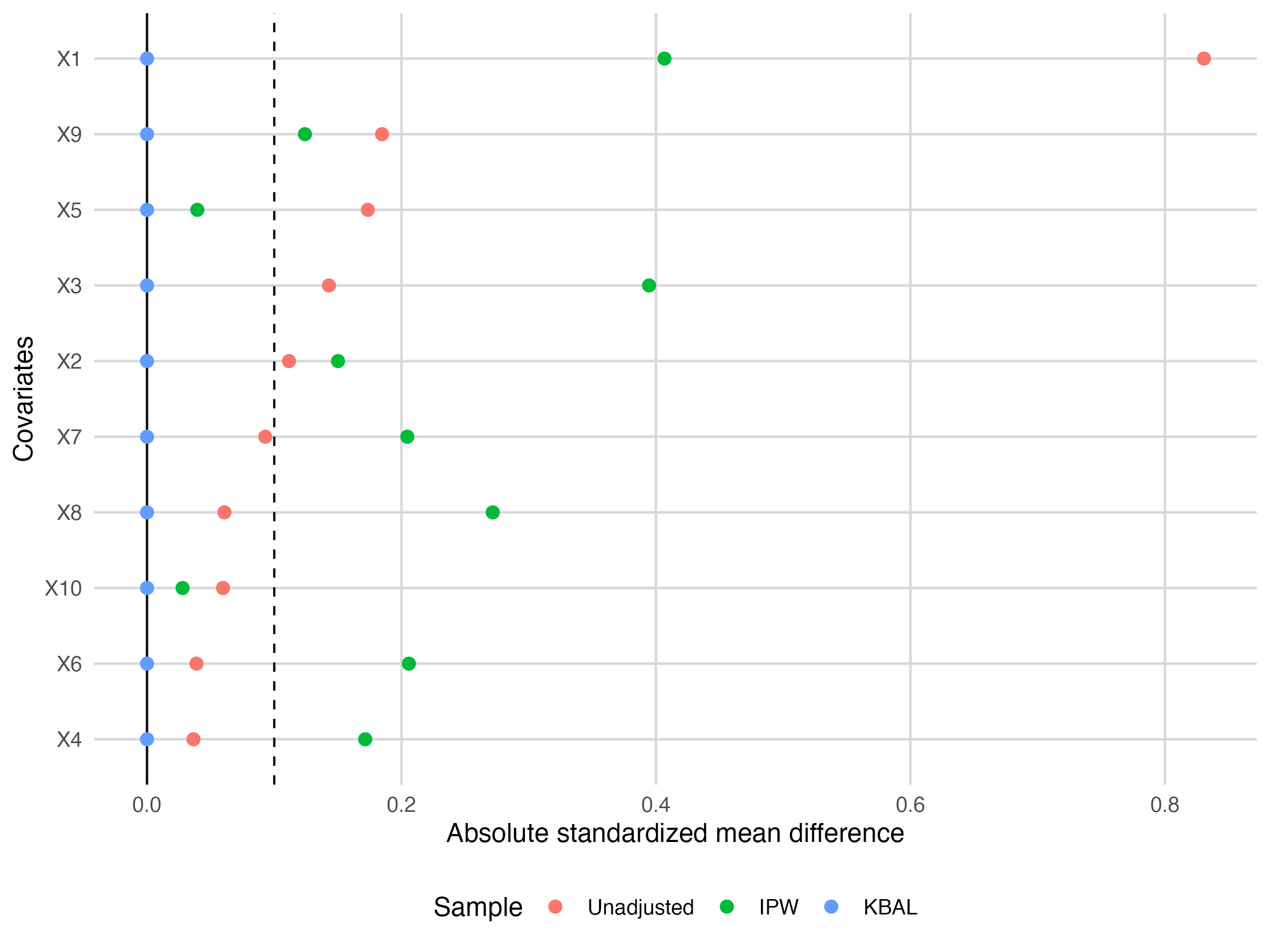}
    \smallskip
    (c) S2: Moderate overlap
\end{minipage}\hfill
\begin{minipage}[t]{0.45\textwidth}
    \centering
    \includegraphics[width=\linewidth]{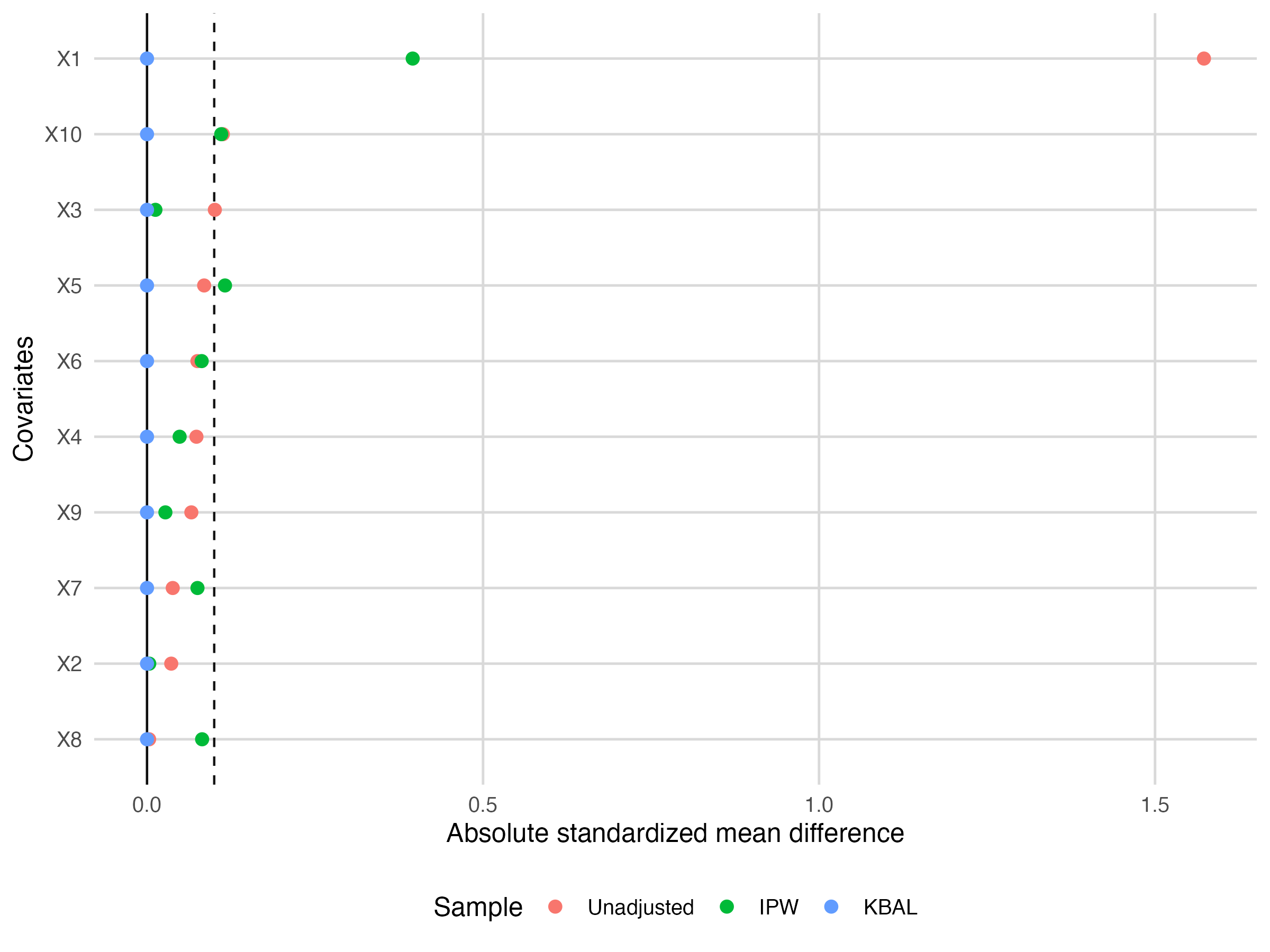}

    \smallskip
    (d) S2: Poor overlap
\end{minipage}

\vspace{0.8em}

\begin{minipage}[t]{0.45\textwidth}
    \centering
    \includegraphics[width=\linewidth]{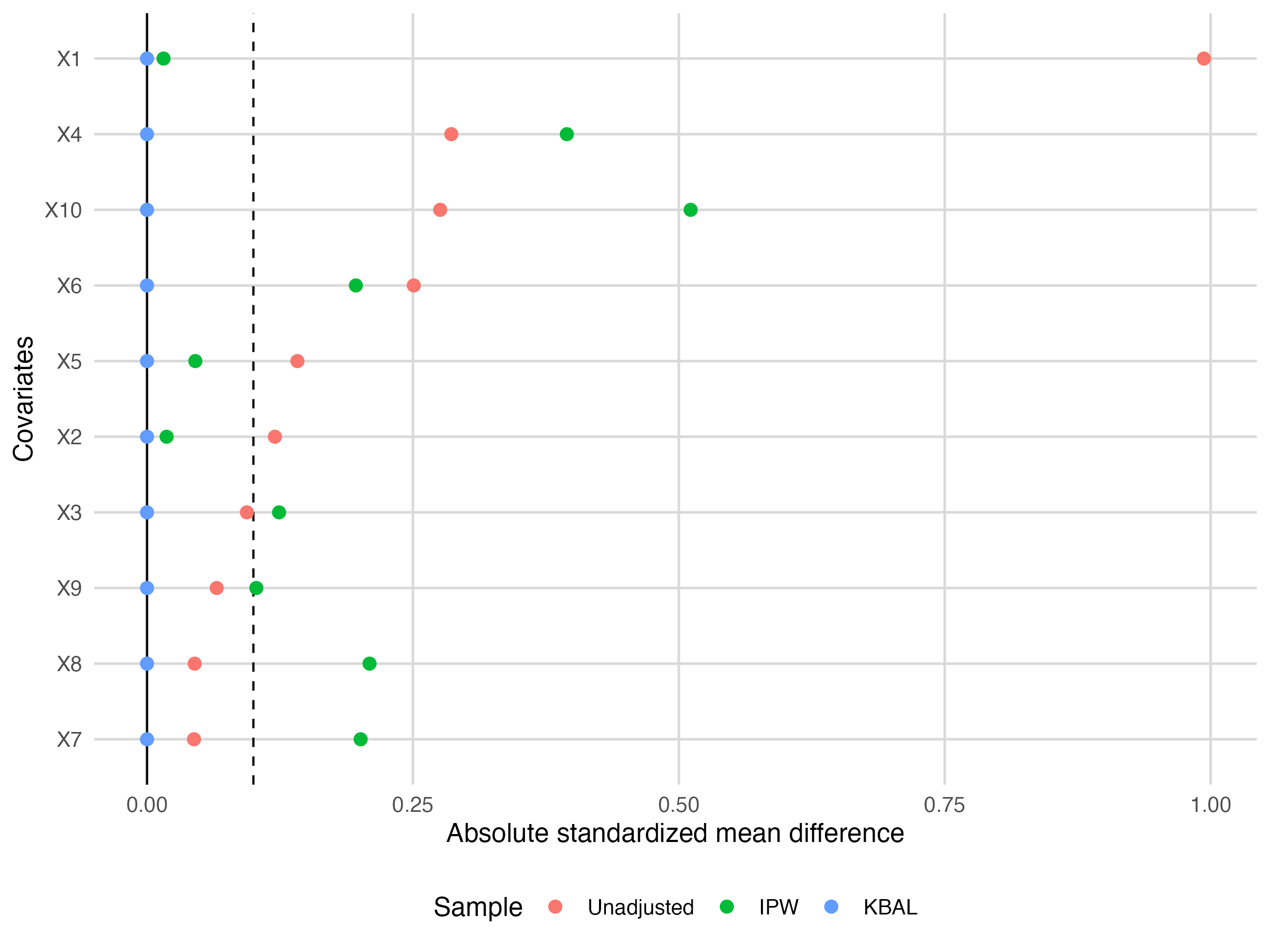}

    \smallskip
   (e) S3: Moderate overlap
\end{minipage}\hfill
\begin{minipage}[t]{0.45\textwidth}
    \centering
    \includegraphics[width=\linewidth]{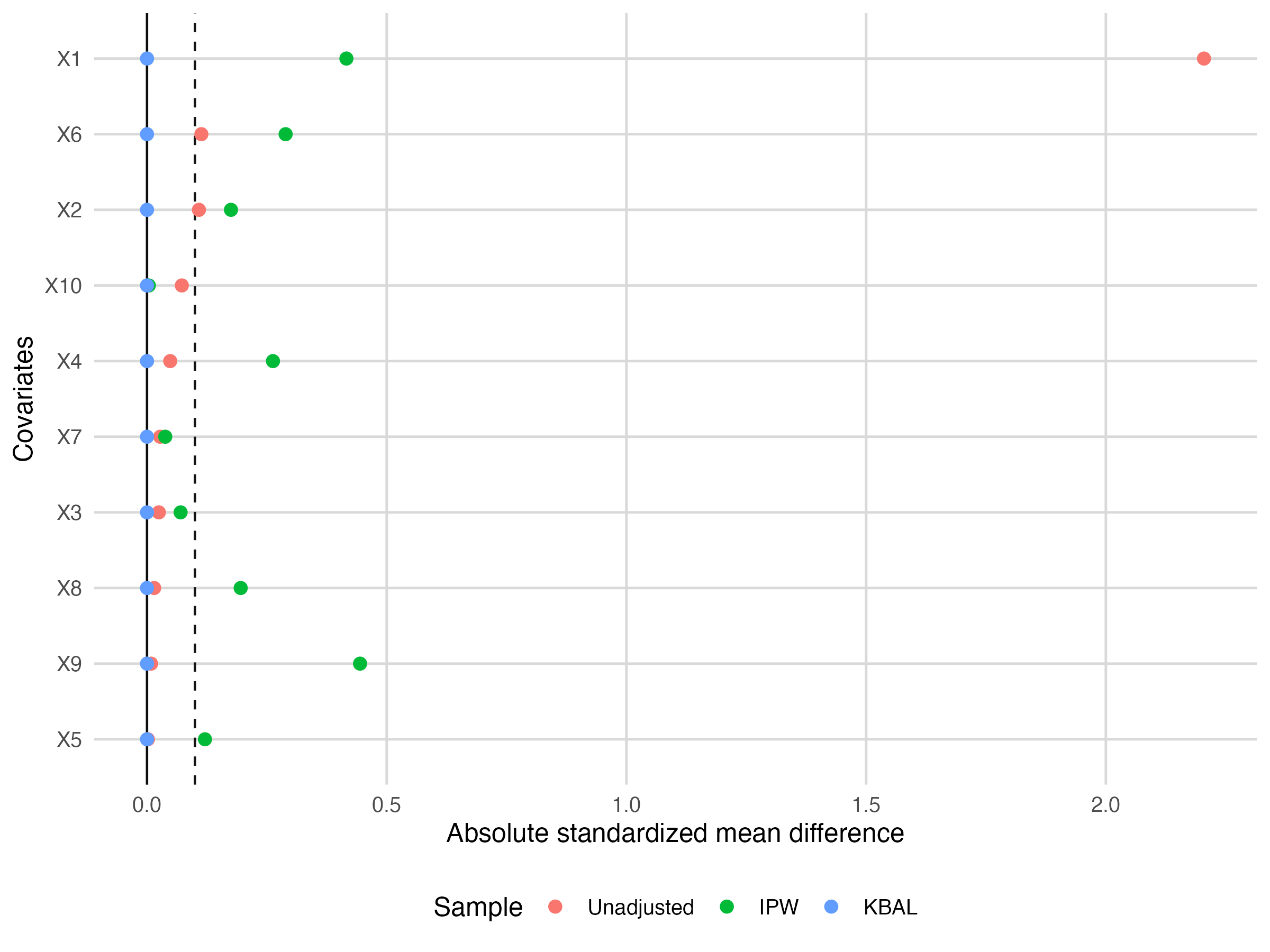}

    \smallskip
    (f) S3: Poor overlap
\end{minipage}

\vspace{0.5em}

\begin{notes}
 Each panel shows the absolute standardized mean differences for covariates before and after applying IPW and KBAL weights. The dashed line marks the balance threshold $|\mathrm{SMD}|=0.1$. KBAL achieves almost exact mean balance, while IPW often results in imbalance, particularly when covariate overlap is moderate. The weights used here correspond to those used in the weighted CF.
\end{notes}
\end{figure}

\FloatBarrier

\begin{figure}[!htbp]
\centering
\caption{Covariate Balance Before and After Weighting}
\label{fig:balance_1x2}
\begin{minipage}[t]{0.48\textwidth}
    \centering
    \includegraphics[width=\linewidth]{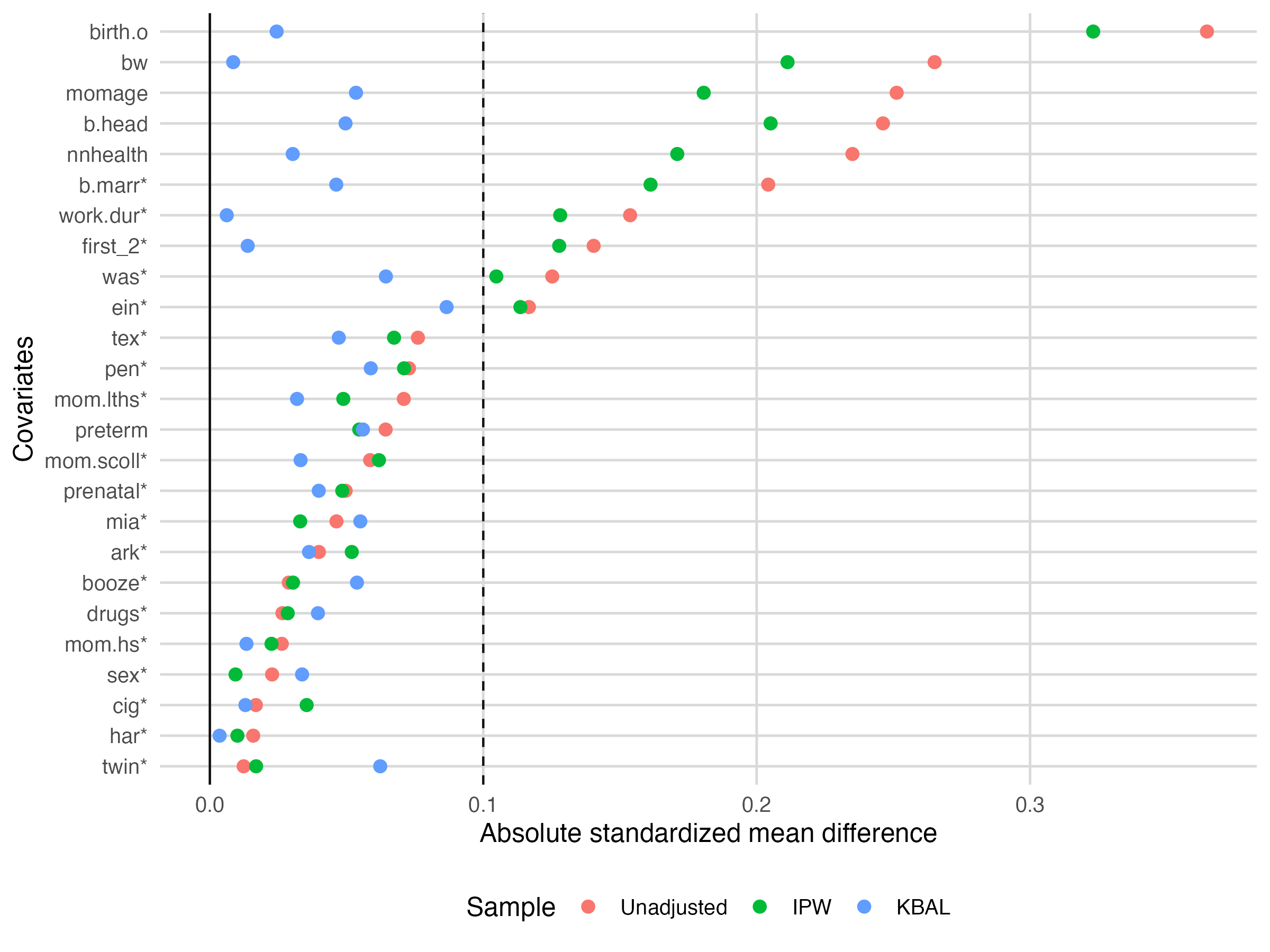}
    \smallskip
    \textbf{(a)} IHDP A
\end{minipage}\hfill
\begin{minipage}[t]{0.48\textwidth}
    \centering
    \includegraphics[width=\linewidth]{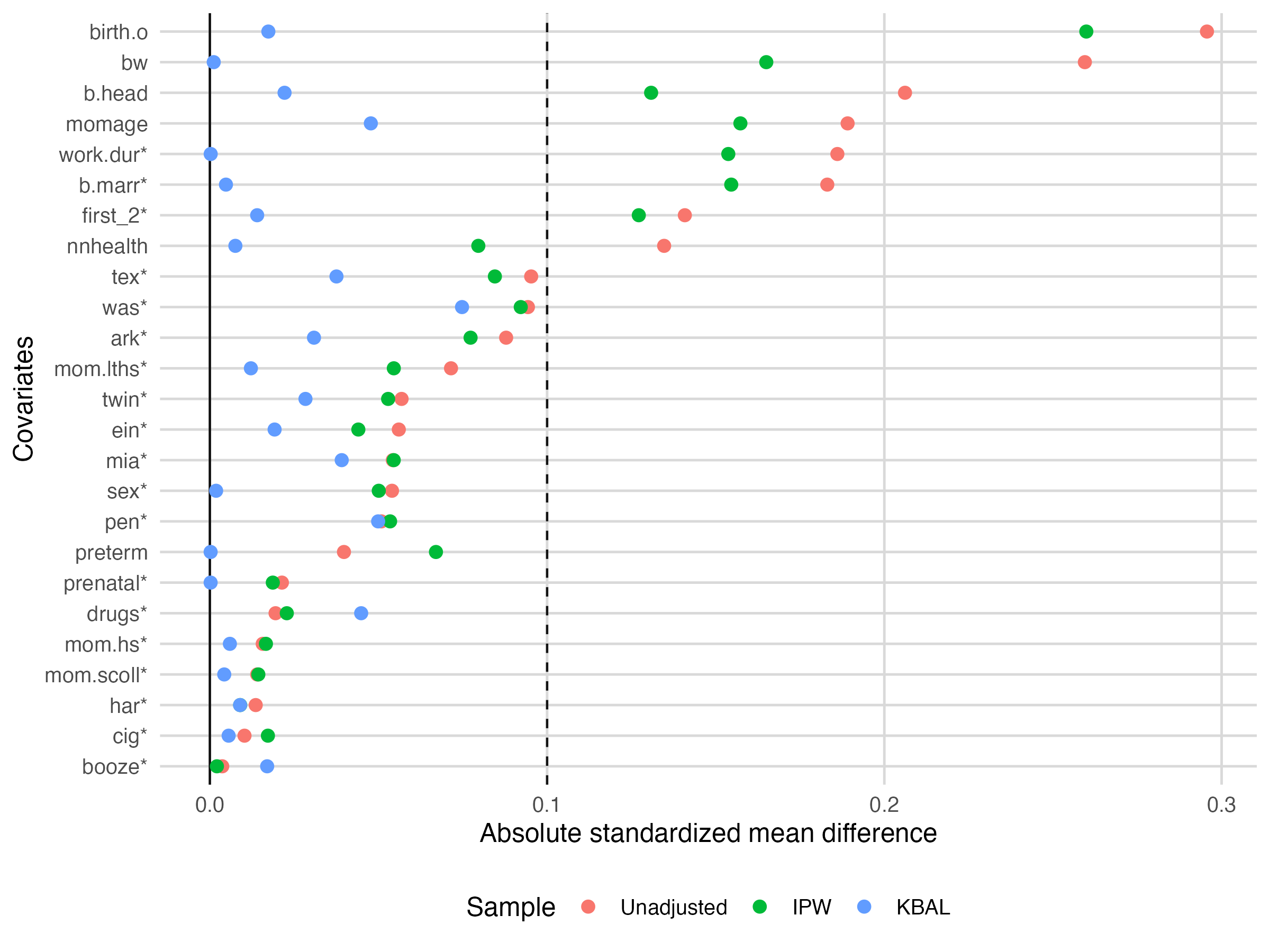}

    \smallskip
    \textbf{(b)} IHDP B
\end{minipage}

\vspace{0.5em}

\begin{notes} Absolute standardized mean differences before and after IPW and KBAL weighting for one IHDP replication. Stars denote binary covariates. Left: IHDP A, right: IHDP B. KBAL reduces covariate imbalance and keeps most covariates below the balance threshold, while unadjusted and IPW show larger imbalance, especially for continuous covariates.
\end{notes}
\end{figure}
\begin{table}[!htbp] 
\centering 
\caption{Ranges of standardized continuous covariates in the IHDP data}
\label{tab:summary_neonatal} 
\begin{tabular}{@{\extracolsep{5pt}}lccc} 
\\[-1.8ex]\hline 
\hline \\[-1.8ex] 
Variable & \multicolumn{1}{c}{N} & \multicolumn{1}{c}{Min} & \multicolumn{1}{c}{Max} \\ 
\hline \\[-1.8ex] 
Birth Order & 747 & -0.88 & 2.24 \\ 
Birth Weight & 747 & -2.73 & 1.51 \\ 
Head Circumference & 747 & -3.80 & 2.60 \\ 
Mother's Age at Birth & 747 & -1.85 & 2.95 \\ 
Neonatal Health Index & 747 & -5.13 & 2.37 \\ 
Weeks Born Preterm & 747 & -1.85 & 2.99 \\ 
\hline \\[-1.8ex] 
\end{tabular} 
\begin{flushleft}
\footnotesize
\textit{Notes:} Entries report the ranges of the standardized continuous covariates used in the IHDP analysis. Since the variables are standardized before analysis, means are approximately zero and standard deviations are one.
\end{flushleft}
\end{table}
\begin{table}[htbp]
\centering
\caption{Summary statistics for binary covariates (IHDP data).}
\label{tab:summary_binary}
\begin{tabular}{@{\extracolsep{5pt}}lccccc} 
\\[-1.8ex]\hline 
\hline \\[-1.8ex] 
Variable & No & Yes \\
\hline
Mother Married at Birth & 358 & 389 \\
Mother Drank Alcohol During Pregnancy & 642 & 105 \\
Mother Smoked During Pregnancy & 479 & 268 \\
Mother Took Drugs During Pregnancy & 30 & 717 \\
Received Prenatal Care & 27 & 720 \\
Child is Male & 363 & 384 \\
Twin Birth Indicator & 677 & 70 \\
Mother Worked During Pregnancy & 303 & 444 \\
\hline
\end{tabular}
\begin{flushleft}
\footnotesize
\textit{Notes:} Entries report the number of observations in each category for the listed binary covariates.
\end{flushleft}
\end{table}
\begin{figure}[H]
  \centering
  \caption{Comparison of propensity score distributions for treated and control groups in the IHDP dataset.}
  \label{fig:propensities_IHDP}
  \includegraphics[width=1\textwidth]{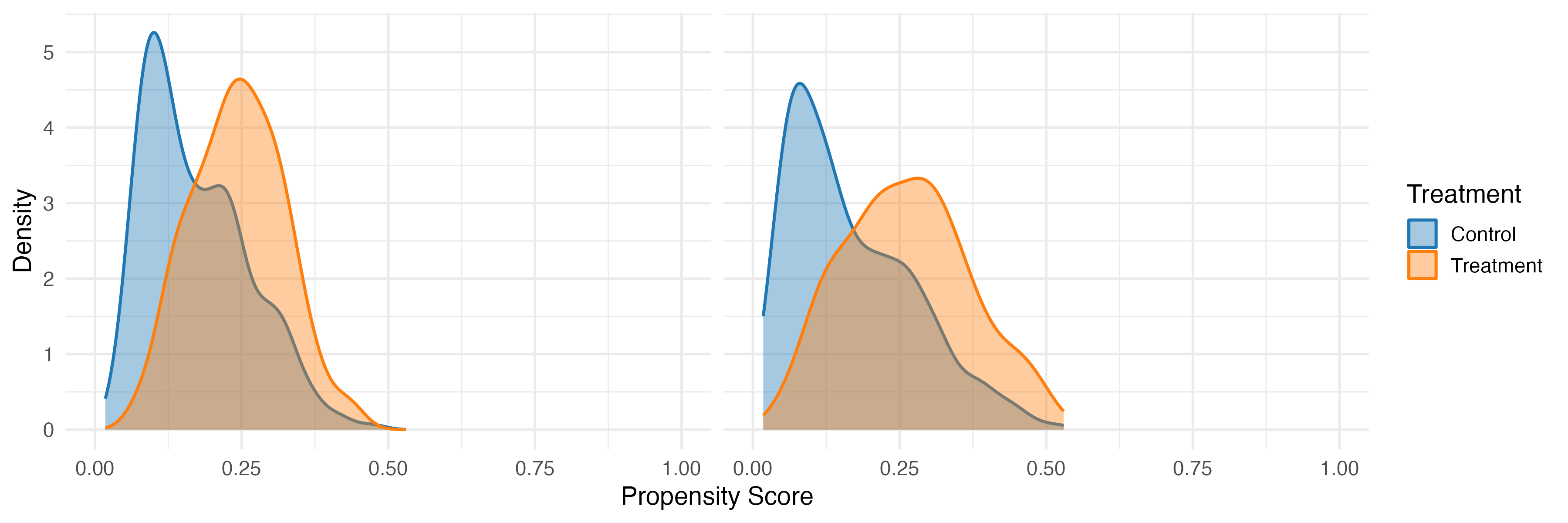}
\caption*{\footnotesize\textit{Notes:}   The left panel corresponds to IHDP Setting A and the right panel to Setting B. Propensity scores are estimated using logistic regression. Limited overlap in the lower and upper tails indicates violations of the positivity assumption, which can lead to unstable propensity score weights and bias in treatment effect estimation.} \end{figure}

\section{Proofs}
\label{Appendix_Proofs}

This section provides the proofs of the theoretical results from the main text. We proceed in the order in which the results are presented. First, we prove Theorem~\ref{thm:DIM_estimator}, showing that the weighted difference-in-means estimator is unbiased for the ATE under the balancing and linearity conditions. We then establish Proposition~\ref{fkbal_conistency}, which shows that $f_{kbal}$ is a suitable weighting function.

\begin{proof}[\textbf{Proof of Theorem \ref{thm:DIM_estimator}}] \label{proof:DIM_estimator}
This proof is based on Theorem 1 of \citet{hazlett2016kernel}, which establishes unbiasedness for the ATT.
We first define the sample average treatment effect (SATE) as
$$\mathrm{SATE} \coloneqq \frac{1}{N}\sum_{i=1}^N (Y_i(1)-Y_i(0)).$$
Since $\mathbb{E}[\mathrm{SATE}] = \mathrm{ATE}$, it suffices to show that $\widehat{\mathrm{DIM}}_w$
is unbiased for $\mathrm{SATE}$. 
Recall the weighted difference-in-means estimator 
\begin{equation*}
\widehat{\mathrm{DIM}}_w =\sum_{i \in \mathcal{S}_1} w_i^1 Y_i^{obs} - \sum_{i \in \mathcal{S}_0} w_i^0 Y_i^{obs}.
\end{equation*}
By SUTVA, $Y_i^{obs}=Y_i(1)$ when $D_i=1$ and $Y_i^{obs}=Y_i(0)$ when $D_i=0$, so
\begin{equation*}
\widehat{\mathrm{DIM}}_w=\sum_{i \in \mathcal{S}_1} w_i^1 Y_i(1)-\sum_{i \in \mathcal{S}_0} w_i^0 Y_i(0).
\end{equation*}
Hence the estimation error is
\begin{equation*}
   \widehat{\text{DIM}}_w - \text{SATE}  
   = \left( \sum_{i \in \mathcal{S}_1} w_i^1 Y_i(1) -  \frac{1}{N} \sum_{i =1}^N Y_i(1) \right) - \left( \sum_{i \in \mathcal{S}_0} w_i^0 Y_i(0) - \frac{1}{N} \sum_{i =1}^N Y_i(0) \right).
\end{equation*}
We treat the two bracketed terms separately. By Assumption~\ref{Linearity of Expected Outcome}, there exist $\theta_1,\theta_0$ such that $$
\mathbb{E}[Y_i(d)\mid X_i=x]=\varphi(x)^\top\theta_d, \quad
Y_i(d) = \varphi(X_i)^\top \theta_d + \varepsilon_i^d \quad \text{for}
\quad
\mathbb{E}[\varepsilon_i^d \mid X_i] = 0, \quad d\in\{0,1\}. 
$$
Hence, 
\begin{align*}
     \sum_{i \in \mathcal{S}_1} w_i^1 Y_i(1) -  \frac{1}{N} \sum_{i =1}^N Y_i(1)  &= \sum w_i^1 \left(  \varphi(X_i)^\top \theta_1 + \varepsilon_i^1\right )- \frac{1}{N} \sum_{i =1}^N \left( \varphi(X_i)^\top \theta_1 + \varepsilon_i^1\right ) \\
    &= \theta_1^\top \left( \sum_{i \in \mathcal{S}_1} w_i^1
    \varphi(X_i)- \frac{1}{N}\sum_{i =1}^N  \varphi(X_i)\right)+\sum_{i \in \mathcal{S}_1} w_i^1 \varepsilon_i^1 - \frac{1}{N}\sum_{i =1}^N \varepsilon_i^1.
  \end{align*}
By the balancing constraint $\sum_{i \in \mathcal{S}_1} w_i^1\varphi(X_i)=\frac1N\sum_{i=1}^N\varphi(X_i)$, the entire $\theta_1$-term is exactly zero. Therefore
\begin{equation*}
         \sum_{i \in \mathcal{S}_1} w_i^1Y_i(1)-\frac1N\sum_{i=1}^N Y_i(1) =\sum_{i \in \mathcal{S}_1} w_i^1\varepsilon_i^1-\frac1N\sum_{i=1}^N\varepsilon_i^1.
         \end{equation*}
The right bracket follows analogously. Combining the two steps gives
\begin{equation*}
         \widehat{\mathrm{DIM}}_w-\mathrm{SATE}=
\left(\sum_{i \in \mathcal{S}_1} w_i^1\varepsilon_i^1-\frac1N\sum_{i=1}^N\varepsilon_i^1\right)
-
\left(\sum_{i \in \mathcal{S}_0} w_i^0\varepsilon_i^0-\frac1N\sum_{i=1}^N\varepsilon_i^0\right).
\end{equation*}
Then, using the Law of Iterated Expectations (LIE) conditioning on $(X_i,D_i)$ we obtain 
\begin{align*}
\mathbb{E}\!\left[\widehat{\mathrm{DIM}}_w-\mathrm{SATE}\right]
&=\mathbb{E}\!\left[\mathbb{E}\!\left[\widehat{\mathrm{DIM}}_w-\mathrm{SATE}\middle|\, X_i,D_i \right]\right] \\
&=\mathbb{E}\!\left[\mathbb{E}\left[\left(\sum_{i \in \mathcal{S}_1} w_i^1\varepsilon_i^1-\frac1N\sum_{i=1}^N\varepsilon_i^1\right)
-\left(\sum_{i \in \mathcal{S}_0} w_i^0\varepsilon_i^0 -\frac1N\sum_{i=1}^N\varepsilon_i^0\right)\,\middle|\, X_i,D_i
\right]\right].
\end{align*}
We first consider the inner conditional expectation.
Since $w^1,w^0$ are measurable functions of the observed covariates and treatment assignments, they are fixed conditional $(X_i,D_i)$. Under Assumption~\ref{KBAL_unconfoundedness} and
$\mathbb{E}[\varepsilon_i^d\mid X_i]=0$ for $d\in\{0,1\}$, we have
\begin{equation*}
\mathbb{E}\!\left[\sum_{i:D_i=d} w_i^d\varepsilon_i^d \,\middle|\, X_i,D_i\right] = \sum_{i:D_i=d} w_i^d\,\mathbb{E} [\varepsilon_i^d\mid X_i]=0,\quad \text{and}\quad 
\mathbb{E}\!\left[\frac1N\sum_{i=1}^N\varepsilon_i^d \,\middle|\, X_i,D_i\right]=0.
\end{equation*}

Therefore,
\begin{equation*}
\mathbb{E}\!\left[\widehat{\mathrm{DIM}}_w-\mathrm{SATE}\middle|\, X_i,D_i \right]=0. 
\end{equation*}
And consequently
\begin{equation*}
\mathbb{E}\!\left[\widehat{\mathrm{DIM}}_w-\mathrm{SATE}\right]=0,
\quad\text{and}\quad
\mathbb{E}\!\left[\widehat{\mathrm{DIM}}_w\right]=\mathbb{E}[\mathrm{SATE}].
\end{equation*}
Finally, $\mathbb{E}[\mathrm{SATE}]=\mathrm{ATE}$ under the data-generating
distribution, hence $\widehat{\mathrm{DIM}}_w$ is unbiased for the ATE.
\end{proof}

\begin{proof}[\textbf{Proof of Proposition~\ref{fkbal_conistency}}]
By assumption,
\[
\hat{\tau}_0(X_i) \xrightarrow{p} \tau(X_i),
\qquad
\hat{\tau}_1(X_i) \xrightarrow{p} \tau(X_i), 
\]
for a fixed unit \(i\).
The aggregated estimator is
\[
\hat{\tau}(X_i)
=
f_{kbal,i}\hat{\tau}_0(X_i)
+
(1-f_{kbal,i})\hat{\tau}_1(X_i).
\]
Subtracting \(\tau(X_i)\) gives
\[
\hat{\tau}(X_i)-\tau(X_i)
=
f_{kbal,i}\left(\hat{\tau}_0(X_i)-\tau(X_i)\right)
+
(1-f_{kbal,i})\left(\hat{\tau}_1(X_i)-\tau(X_i)\right).
\]
Since \(f_{kbal,i}\in[0,1]\), both coefficients are bounded. Hence, by Slutsky's theorem,
\[
f_{kbal,i}\left(\hat{\tau}_0(X_i)-\tau(X_i)\right)
\xrightarrow{p}0,
\qquad
(1-f_{kbal,i})\left(\hat{\tau}_1(X_i)-\tau(X_i)\right)
\xrightarrow{p}0.
\]
Therefore,
$\hat{\tau}(X_i)-\tau(X_i)\xrightarrow{p}0, $
which implies $\hat{\tau}(X_i)\xrightarrow{p}\tau(X_i).$
\end{proof}

\section{Implementation in \texttt{R}} \label{Appendix_Implementation}

All simulations are conducted in $\texttt{R}$ using RStudio as the
development environment \citep{r_2025, rstudio_2025}. The code for reproducibility is provided at
\url{https://github.com/karo93/KBalweighting_Trees}.  

For the Monte Carlo simulation study we use the packages $\texttt{grf}$ \citep{grf_package}, $\texttt{Rforestry}$ \citep{kuenzel2025rforestry} and $\texttt{kbal}$ \citep{hazlett2026kbal}, along with custom extensions based on $\texttt{causalToolbox}$ \citep{kuenzel_causaltoolbox}.
Causal forests are estimated using the \texttt{grf} package. Metalearners are implemented using the $\texttt{causalToolbox}$ package,  with random forest base learners from \texttt{Rforestry}.
In each Monte Carlo run, we draw a training sample used for model fitting and a test sample used for evaluation. Kernel balancing weights are computed on the covariates of the sample on which weighting is applied: training covariates for weighted causal forests and test covariates for the X-learner aggregation weights.

\subsection*{Choice of Parameters}
First we provide an overview of the \texttt{kbal()} function, which is part of the \texttt{kbal} package. The \texttt{kbal()} function constructs optimal weights based on the sets $\mathcal{U}$ and $\mathcal{V}$ described in Table~\ref{table:weight_sample}. Below, we summarize the main tuning parameters used in the \texttt{kbal} function for computing the optimal weights:

\begin{itemize}
    \item \texttt{allx}: Data matrix of covariates. Covariates from the training set (\texttt{X}) are used to compute causal forest weights, while covariates from the test set (\texttt{Xtest}) are used to compute weights for the X-Learner.
\item \texttt{sampled}: Numeric vector of length \(N\) identifying the units to be reweighted. Units with \(\texttt{sampled}=1\) form the reweighted group, while units with \(\texttt{sampled}=0\) form the target group. In our implementation, it is constructed from the treatment indicator.
\item \texttt{sampledinpop}: Logical. Must be set to \texttt{TRUE} to compute the correct weights (see Table~\ref{table:weight_sample}).
    \item \texttt{b}: Scaling factor for Gaussian kernel distances; here set to \texttt{length(allx)}. 
    \item \texttt{linkernel}: Logical, set in Monte Carlo study \texttt{TRUE} to compute kernel balancing weights with a linear kernel $K = XX^\top$, improving computational efficiency while achieving mean balance on the first moments of $X$.
\end{itemize}
\subsection*{Weighted Causal Forest}
The \texttt{grf} package \citep{grf_package} allows to specify observation-level weights via the
\texttt{sample.weights} argument. These weights are interpreted as describing a target population in which observation $X_i$ is sampled with probability proportional to \texttt{sample.weights[i]} and should not be confused with the forest similarity weights
$\alpha_i(x)$ used in prediction (see Algorithm~\ref{algo:grf_ct}). By default, all observations receive equal weight, so that the target population coincides with the empirical distribution of the training data. For causal validity, the sample weights should be constructed independently of the observed outcomes, conditional on the covariates and treatment assignment. Kernel-balancing weights satisfy this requirement, since they are functions of the covariates and treatment indicators, \(X_i\) and \(D_i\), but not of the observed outcomes \(Y_i\) \citep{hazlett2016kernel}.
We compute these on the training sample and pass these to the \texttt{grf::causal\_forest} function via the
\texttt{sample.weights} argument. For conditional estimands such as \(\tau(x)\), sample weights do not redefine the estimand itself. Instead, they influence how the forest is trained: observations with larger weights receive more importance in the splitting and estimation procedure, so the fitted forest emphasizes regions of the covariate distribution that receive larger weight. In our setting, KBal weights are used as an initial global balancing step to make the treated and control samples more comparable in the training data before the causal forest learns the local partitioning structure and estimates heterogeneous treatment effects.

This should be distinguished from a fully local balancing approach, in which balancing weights would be recomputed within each leaf or neighborhood. Such an approach would more directly target covariate-specific balance, but is not how sample weights are incorporated in the standard \texttt{grf} implementation. All other tuning parameters of the causal forest are kept at their default values.

\subsection*{Weighted X-Learner}
For the simulation study, we adapted the XRF function from the \texttt{causalToolbox} package rather than relying on a general wrapper with dynamically passed hyperparameters. This was necessary because the standard implementation uses the estimated propensity score $\hat{p}(x)$ in the final aggregation step, whereas our weighted specification replaces this aggregation weight with the KBal-based weight described in Section~\ref{sec:Weighted_XF}. The first- and second-stage regression steps estimate the response functions $\hat{\mu}_0(x)$ and $\hat{\mu}_1(x)$ and the imputed treatment effect functions
$\hat{\tau}_0(x)$ and $\hat{\tau}_1(x)$ are implemented exactly as in the original XRF implementation. In particular, we use identical tuning parameters, including subsampling rates and honesty settings. This ensures comparability across both methods, because differences between the weighted and unweighted X-learner arise solely from the choice of aggregation function, while allowing controlled changes to the weighting function. 

As discussed in Section~\ref{sec:Kernel_weighting}, the KBal aggregation weights are sample-specific, because a closed-form weighting function is not available. To avoid dimension mismatches between the training and test samples, the KBal aggregation weights are computed directly with \(X_{\text{test}}\). Since the weights only enter the final aggregation step and do not use evaluation outcomes, this does not involve outcome leakage. The fitted components \(\hat{\tau}_0(x)\) and \(\hat{\tau}_1(x)\) are trained only on the training sample, while evaluation outcomes are used solely for performance assessment.

Furthermore, the notation $f_{kbal,i}$ may appear to create a dimension mismatch, since the weights $w^0$ and $w^1$ are obtained from separate balancing problems for the treated and control groups. However, in the implementation, the \texttt{kbal} function from the kernel balancing package \citep{hazlett2026kbal} returns both weight vectors on the same evaluation sample. Observations belonging to the group that defines the target distribution receive weight one, while the remaining observations receive the optimized balancing weights. Hence, $w_i^0$ and $w_i^1$ are both defined in the implementation for each observation $i \in \{1,\dots,N\}$ and therefore
 \[
    f_{kbal,i} = \frac{w_i^{0}}{w_i^{0}+w_i^{1}}
    \]
is computed observation by observation.
\subsection{Algorithmic Details} \label{app:algorithm}

\subsection*{Generalized Causal Forest}
The original GRF framework targets parameters of the form $\theta(x) = \xi^\top \beta(x)$, where $\beta(x)$ is a vector of local regression coefficients, and $\xi $ is a contrast vector. These parameters are identified via the local moment condition:
\begin{equation}
    \mathbb{E} \left[ \left(Y_i - D_i^\top \beta(x) - c(x) \right)(1, D_i)^\top \;\middle|\; X_i = x \right] = 0,
\end{equation}
where \( c(x) \) is an intercept term. The contrast vector \( \xi \) determines which component of \( \beta(x) \) is being targeted. In our case, where the treatment variable \( D_i \) is binary, we focus on estimating \( \tau(x) \). This corresponds to setting \( \xi = (0,1)^\top \), so that \( \theta(x) = \tau(x) \). Algorithm \ref{algo:grf_ct} and \ref{algo:gradienttree} are a rewritten and simplified version of the GRF framework, where we focus just on $\tau(x)$. Note we used the GRF formulation in its original vector-valued notation, so here, \(D_i\) may be vector-valued. 

  \begin{algorithm}[H]
        \SetAlgoLined
        \KwIn{$N$ training examples $(X_i, Y_i, D_i)$}
        \KwOut{CATE estimator $\hat{\tau}(x)$}
        \For{$b = 1, \dots, B$}{
            \textbf{Step 1:} Draw a random subsample of size $s$ from $\{1, \dots, n\}$ without replacement, 
    and then divide it into two disjoint sets of size $|\cal{I}|$ $ = \left\lfloor \frac{s}{2} \right   \rfloor$ and $|\cal{J}|$ $= \left\lceil \frac{s}{2} \right\rceil$  \\[0.5cm]
            \textbf{Step 2:} Grow a gradient tree ${\cal{T}}_b$ using the $\cal{I}$ sample.   \\[0.5cm]
            \textbf{Step 3:} Given a test point $x$,  assign samples in $\cal{J}$ to leaf nodes based on ${\cal{T}}_b$. Compute
           $$ \alpha_{bi}(x) = \frac{\mathds{1}(\{X_i \in L_b(x)\})} {|L_b(x)|}.$$
             \\[0.5cm]
        }
        \textbf{Final Estimate:} Compute the weighted estimator
        \[
        \hat{\tau}(x) = \left( \sum_{i=1}^n \alpha_i(x) (D_i - D_\alpha)(D_i - D_\alpha)^\top \right)^{-1} \sum_{i=1}^n \alpha_i(x) (D_i - D_\alpha)(Y_i - Y_\alpha). 
        \]
        \caption{Generalized Causal Forest Algorithm}
\label{algo:grf_ct}
    \end{algorithm}
\FloatBarrier 

\FloatBarrier 
\vspace{0.5cm}
  \begin{algorithm}[h]
        \SetAlgoLined
            \textbf{Step 1:}  For a parent node $P$, compute $A_P$
\[
A_P = \frac{1}{|\{i : X_i \in P\}|} \sum_{\{i : X_i \in P\}} (D_i - \overline{D}_P)(D_i - \overline{D}_P)^\top, 
\] 
where $\overline{D}_P$  is the average taken over the parent $P$. 
  \\[0.5cm]
\textbf{Step 2:} Compute the pseudo-outcomes \(\rho_i\)  for each observation \(i\) with \(X_i \in P\) \[
\rho_i = A_P^{-1} (D_i - \overline{D}_P) \left( Y_i - \overline{Y}_P - (D_i - \overline{D}_P)\hat{\tau}_P \right),
\]
where $\hat{\tau}_P $ is the least-squares regression solution of $Y_i$ on $D_i$ and $ \overline{Y}_P$ is the average taken over the parent $P$.  \\[0.5cm]

            \textbf{Step 3:} Split $P$ into two children $C_1$ and $C_2$ such as to maximize the criterion $$\tilde{S}(C_1, C_2) = \sum \limits_{j=1}^2 \frac{1}{\left| \{i : X_i \in C_j\} \right|} \left( \sum \limits_{\{i : X_i \in C_j\}} \rho_i \right)^2. $$
        
        \caption{Gradient Tree ${\cal{T}}_b$}
\label{algo:gradienttree}
    \end{algorithm}
\vspace{0.5cm}
\FloatBarrier 

\subsection*{Metalearners} \label{app:metalearners}
One of the simplest forms of a metalearner is the S-learner, where the letter S stands for “single” because it fits a single model to estimate the CATE by including the treatment variable as a covariate
\[
\hat{\mu}(x, d) = \mathbb{E}[Y_i^{obs} \mid X_i = x, D_i=d].
\]
The estimated CATE is then obtained as the difference between predicted outcomes under treatment and control.
\[
\hat{\tau}_S(x) = \hat{\mu}(x, D_i=1) - \hat{\mu}(x, D_i=0).
\]
The corresponding procedure is summarized in Algorithm~\ref{algo:S_Learner}. The S-learner is preferable when treatment effects are expected to be small or structurally simple, and when the response functions for treated and control group have similar structure \citep{Comparing_Meta_Learners}. However, because many machine learning methods regularize and may downweight weak predictors, the S-learner can understate treatment effect heterogeneity when the treatment indicator has limited predictive power or when treatment effects are highly nonlinear.

\begin{algorithm}[!htbp]
\SetAlgoLined
\KwIn{Observations $(X_i, Y_i^{obs}, D_i)_{i=1}^N$}

\textbf{Step 1:} Estimate the conditional mean function
\[
\hat{\mu}(x,d) \approx \mathbb{E}[Y_i^{obs} \mid X_i=x, D_i=d]
\]
using a supervised learning method with $(X_i, D_i)$ as predictors.

\textbf{Step 2:} Estimate the CATE as
\[
\hat{\tau}_S(x) = \hat{\mu}(x,1) - \hat{\mu}(x,0).
\]
\caption{S-learner}
\label{algo:S_Learner}
\end{algorithm}

 The T-learner (T short for “Two”) is similar to the S-learner, but differs in that it uses the treatment indicator to split the estimation of the response function into two separate models.  Specifically, the treatment and control response functions, \begin{equation} \label{eq:response}
    \mu_1(x) = \mathbb{E}[Y_i(1) \mid X_i = x]~ \text{ and }~ \mu_0(x) = \mathbb{E}[Y_i(0) \mid X_i = x],  \notag
\end{equation} are each estimated independently by (not necessarily identical) base learners, using only the observations from the treated and control groups, respectively. Then the T-learner is defined as
\[
\hat{\tau}_T(x) \coloneqq \hat{\mu}_1(x) - \hat{\mu}_0(x).
\]
The main steps of the T-learner are given in Algorithm~\ref{algo:tlearner}. The T-learner is preferable in settings, where the response functions
$\mu_1(x)$ and $\mu_0(x)$ differ substantially in functional form, sparsity, or
smoothness, or when observational selection induces group-specific
covariate-outcome relationships. A limitation is that the two models do not share information across groups, hence, performance may degrade when treated and control units follow similar response surfaces, like in well-randomized settings. The method is also sensitive to sample imbalance: when $N_1 \ll N_0$
(or vice versa), the response function estimated for the smaller group may overfit,
leading to noisy CATE estimates.

\begin{algorithm}[!htbp]
\SetAlgoLined
\KwIn{Observations $(X_i, Y_i, D_i)_{i=1}^N$}

\textbf{Step 1:} Using only treated observations $(D_i=1)$, fit a supervised learning model to estimate
\[
\hat{\mu}_1(x) \approx \mathbb{E}[Y_i^{obs} \mid X_i=x, D_i=1].
\]

\textbf{Step 2:} Using only control observations $(D_i=0)$, fit a supervised learning model to estimate
\[
\hat{\mu}_0(x) \approx \mathbb{E}[Y_i^{obs} \mid X_i=x, D_i=0].
\]

\textbf{Step 3:} Estimate the CATE as
\[
\hat{\tau}_T(x) = \hat{\mu}_1(x) - \hat{\mu}_0(x).
\]

\caption{T-learner}
\label{algo:tlearner}
\end{algorithm}

For metalearners, pointwise uncertainty is estimated using a stratified nonparametric bootstrap. Algorithm~\ref{algo:bootstrap_meta} describes the resulting procedure.  In each bootstrap replication, treatment and control observations are resampled separately, the corresponding weighted or unweighted X-learner is refitted, and the empirical standard deviation of the resulting CATE predictions is used as the pointwise standard error. As discussed in Appendix~\ref{app:Appendix_monte}, confidence intervals are constructed using a normal approximation centered at \(\hat{\tau}(x)\).

\begin{algorithm}[!htbp]
\SetAlgoLined
\KwIn{
$N$ training observations $(X_i, D_i, Y_i)$}
\KwOut{Pointwise $(1-\alpha)$ confidence interval for $\tau(x)$}

\textbf{Step 1:}  Fit the metalearner on the full training sample and obtain $\hat{\tau}(x)$. \\
\For{$r = 1, \dots, R$}{
    \textbf{Step 2:} Draw stratified bootstrap samples $\mathcal{S}_1^{(r)}$ and $\mathcal{S}_0^{(r)}$ with replacement:
  
    \textbf{Step 3:} Construct the bootstrap dataset
    \[
    \mathcal{D}^{(r)} = \{(X_i, D_i, Y_i) : i \in \mathcal{S}_1^{(r)} \cup \mathcal{S}_0^{(r)}\}
    \]
  and refit the corresponding weighted or unweighted X-learner on $\mathcal{D}^{(r)}$ and compute the bootstrap prediction $\hat{\tau}^{*}_r(x)$.
}

\textbf{Step 4:} Estimate the pointwise standard error using the empirical standard deviation
\[
\widehat{\sigma}(\hat{\tau}(x)) =
\mathrm{sd}\!\left( \hat{\tau}^{*}_1(x), \dots, \hat{\tau}^{*}_R(x) \right).
\]
\textbf{Final Output:} Construct the $(1-\alpha)$ confidence interval
\[
\hat{\tau}(x) \pm z_{1-\alpha/2} \, \widehat{\sigma}(\hat{\tau}(x)).
\]
\caption{Bootstrap Confidence Intervals for Metalearners}
\label{algo:bootstrap_meta}
\end{algorithm}

\end{document}